\documentclass[a4paper, accepted=2023-10-24, onecolumn, 11pt]{quantumarticle}
\pdfoutput=1
\usepackage[utf8]{inputenc}
\usepackage[english]{babel}
\usepackage[T1]{fontenc}
\usepackage[numbers]{natbib}
\usepackage{listings}

\usepackage[resetlabels]{multibib}

\newcites{article}{article references}
\newcites{book}{book references}
\newcites{misc}{misc references}
\newcites{repo}{repository references}
\newcites{web}{website references}
\newcites{other}{Other references}

\usepackage{amsmath, amsfonts, amssymb, bbm, mathtools, verbatim, parskip, xcolor,amsthm}
\usepackage{hyperref}
\usepackage{enumerate}
\usepackage{paralist}
\usepackage[framemethod=tikz]{mdframed}
\usepackage{graphicx}
\usepackage{tikz-cd}
\usepackage{soul, xcolor}
\usepackage[ruled, linesnumbered]{algorithm2e}
\usepackage[margin=1in]{geometry}
\usepackage{amsfonts}
\usepackage{imakeidx}
\newtheorem{thm}{Theorem}
\newtheorem{claim}{Claim}
\newtheorem{fact}[thm]{Fact}
\newcounter{example}[section]
\newenvironment{example}[1][]{\refstepcounter{example}\par\medskip
   \noindent \textbf{Example~\theexample. #1} \rmfamily}{\medskip}
\newtheorem{cor}[thm]{Corollary}

\newtheorem{definition}{Definition}[section]
\newtheorem{theorem}{Theorem}
\newtheorem{lemma}{Lemma}

 \newenvironment{changemargin}[2]{%
\begin{list}{}{%
\setlength{\topsep}{0pt}%
\setlength{\leftmargin}{#1}%
\setlength{\rightmargin}{#2}%
\setlength{\listparindent}{\parindent}%
\setlength{\itemindent}{\parindent}%
\setlength{\parsep}{\parskip}%
}%
\item[]}{\end{list}}

\newtheorem{observation}{Observation}

\newcommand\supp{\mathop{\rm supp}}

\newcommand{\clE}{\mathcal{E}}
\newcommand{\cIF}{\mathcal{F}}

\newcommand{\cH}{\mathcal{H}}

\newcommand{\clO}{\mathcal{O}}

\newcommand{\sfV}{\mathsf{Val}}

\newcommand{\Id}{\mathbbm{1}}
\newcommand{\id}{\mathbb{I}}

\newcommand{\bbC}{\mathbb{C}}
\newcommand{\fix}{\mathsf{Fix}}

\newcommand{\tab}{\hspace{2cm}}

\newcommand{\eps}{\epsilon}
\newcommand{\poly}{\mathrm{poly}}

\newcommand{\hmin}{\mathrm{H}_{\min}}

\newcommand{\ket}[1]{|#1\rangle}
\newcommand{\bra}[1]{\langle#1|}
\newcommand{\ketbra}[2]{|#1\rangle\langle#2|}

\newcommand{\enc}{\mathsf{Enc}}
\newcommand{\dec}{\mathsf{Dec}}

\newcommand{\Tr}{\mathrm{Tr}}

\newcommand{\suppress}[1]{}

\newcommand{\wg}{\mathrm{Wg}}
\allowdisplaybreaks

\usepackage{crossreftools}

\makeatletter
\newcommand{\optionaldesc}[2]{%
  \phantomsection
  #1\protected@edef\@currentlabel{#1}\label{#2}%
}
\makeatother

\DeclareMathOperator{\E}{\mathbf{E}}
\begin{document}
\title{Tamper Detection Against Unitary Operators}

\author{Naresh Goud Boddu}
\affiliation{NTT Research, Sunnyvale, USA}
\author{Upendra Kapshikar}
\affiliation{Center for Quantum Technologies, National University of Singapore, Singapore}

\maketitle

	\begin{abstract}
 Security of a storage device against a tampering adversary has been a well-studied topic in classical cryptography.
 Such models give black-box access to an adversary, and the aim is to protect the stored message or abort the protocol if there is any tampering.
 The study of these models has led to some important cryptographic and communication primitives, such as tamper detection codes and non-malleable codes.
 In this work, we extend the scope of the theory of tamper detection codes against an adversary with quantum capabilities. 
	We consider encoding and decoding schemes that are used to encode a $k$-qubit quantum message $\vert m\rangle$ to obtain an $n$-qubit quantum codeword $\ket{\psi_m}$.
  A quantum codeword $ \ket{\psi_m}$ can be adversarially tampered via a unitary $U$ from some known tampering unitary
		family $\mathcal{U}_{\mathsf{Adv}}$ (acting on $\mathbb{C}^{2^n}$), resulting in $U \ket{\psi_m} \bra{\psi_m} U^\dagger$.
		Firstly, we initiate the general study of \emph{quantum tamper detection codes}, which detect if there is any tampering caused by the action of a unitary operator. 
		In case there was no tampering, we would like to output the original message.
		 We show that quantum tamper detection codes exist for any family of unitary operators $\mathcal{U}_{\mathsf{Adv}}$, such that $\vert\mathcal{U}_{\mathsf{Adv}} \vert < 2^{2^{\alpha n}}$  for some constant $\alpha \in (0,1/6)$; provided that unitary operators satisfy one additional condition: 
		\begin{itemize}
			\item[-] Far from the identity: for each $U \in \mathcal{U}_{\mathsf{Adv}}$, we require that its inner product with the identity operator is not too big, that is, $\vert \left\langle \mathbb{I} , U \right\rangle \vert  =  \vert \Tr(U) \vert \leq \phi 2^n$, where $\phi$ is suitably chosen parameter.
		\end{itemize}
   \vspace{2mm}  
		Quantum tamper detection codes that we construct can be considered to be quantum variants of  \emph{classical tamper detection codes} studied by Jafargholi and Wichs~['15], which are also known to exist under similar restrictions. 
 Additionally, we show that when the message set $\mathcal{M}$ is classical, such a construction can be realized as a \emph{non-malleable code} against an adversary having access to any $\mathcal{U}_{\mathsf{Adv}}$ of size up to $2^{2^{\alpha n}}$. 
	\end{abstract}

\section{Introduction} \label{sec:intro}	
Traditionally, cryptographic schemes have been analyzed assuming that an adversary has only black-box access to the underlying functionality and no way to manipulate the internal state. \emph{Tamper-resilient cryptography} is a model in cryptography where an adversary is allowed to  tamper  with the internal state of a device (without necessarily knowing what it is) and then observe the outputs of the tampered device.
 By doing so, an attacker may learn additional sensitive information that would not be available otherwise. One natural approach to protect against such attacks is to encode the data on the device in some way.
 One can try to use error-correcting codes such as Reed-Solomon codes, but such an encoding will prevent tampering with bounded Hamming weights, typically less than the distance of codes.
Tamper detection codes introduced by Jafargholi and Wichs~\cite{JW15} provide meaningful guarantees on the integrity of an encoded message in the presence of a tampering adversary, even in settings where error correction and error detection may not be possible. 
	
	Consider the following: suppose one wants to store a message in a database accessible to an adversary.
		 The adversary is then allowed to tamper the stored message using a function $f$  from some function family $\cIF_{\mathsf{Adv}}$. Naturally, from a decent storage, we expect two properties~-
\begin{itemize}[•]
\item If there is tampering, we should be able to detect it with high probability.
\item If there was no tampering, then we should always be  able to recover the original message. 
\end{itemize}

Let $\mathcal{M}$ be the set of messages, and let the storage be labelled by $\mathcal{C}$.
For such a scheme, we require an encoder (\textsf{Enc}) from $\mathcal{M}$ to $\mathcal{C}$ and a decoding procedure (\textsf{Dec}) that reverses this operation.
The decoder \textsf{Dec}  is additionally allowed to output a special symbol $\perp$, to indicate that the message was tampered.
The experiment can be modelled as a simple three-step procedure:
		\begin{enumerate}[a)]
			\item  A message $m \in \mathcal{M}$ is encoded via a (possibly randomized) encoder $\enc: \mathcal{M} \rightarrow \mathcal{C}$, yielding a codeword $c =  \enc (m).$ 
			\item An adversary can tamper $c$ (non-trivially) via a  function $f$ from some known tampering function family $\cIF_{\mathsf{Adv}}$, resulting in $\hat{c}=f(c)$.
			\item The tampered codeword $\hat{c}$ is then decoded to a candidate message $\hat{m}\in \mathcal{M} \cup \lbrace \perp \rbrace$ using a (possibly randomized) decoder $\dec : \mathcal{C} \to \mathcal{M} \cup \{ \perp \}$.
		\end{enumerate} 
The properties that we desire from this scheme are:
\begin{enumerate}[A.]
\item \label{item:intro_completeness_1} $\Pr\left(\dec\left(\enc(m)\right)=m\right)=1$ (\emph{Completeness}).	
\item \label{item:intro_soundness_1} $\Pr\left(\dec \left( f \left( \enc(m) \right) \right)= \perp \right) \geq 1-\eps$ (\emph{Soundness}).	
\end{enumerate} 	
Property~\ref{item:intro_completeness_1} indicates that if no one tampers anything, we can always get back the original message. Property~\ref{item:intro_soundness_1} states that the decoder can detect every non-trivial tampering with probability~\footnote{The probability stated above is taken over the randomness of the encoder and decoder.}  at least $1-\eps$.
If some encoding and decoding scheme (\textsf{Enc, Dec}) satisfies the above properties, we say that it is an $\eps$-secure tamper detection code (for family $\cIF_{\mathsf{Adv}}$).

Note that Property~\ref{item:intro_soundness_1} can hold in two different degrees.
One, it is valid for all messages $m$; where we call the scheme to be a strong tamper detection code (or simply tamper detection code).
And two, it can be valid for a randomly chosen $m$, in which case we say the scheme is a weak tamper detection code. 
In this work, we restrict ourselves to the strong form of tamper detection.

For tampering to be meaningful, we assume that $f$ is not the identity map.  
It is easy to see that for any function family $\cIF_{\mathsf{Adv}}$, the storage size $\vert \mathcal{C} \vert$ has to be greater than or equal to $\vert \mathcal{M} \vert$. Otherwise, the encoding scheme will be many-to-one, and Property~\ref{item:intro_completeness_1} can not be satisfied. 
Also, the larger the family $\cIF_{\mathsf{Adv}}$ becomes, the stronger the adversary gets, and we expect the size of $\mathcal{C}$ to increase.
This raises a natural question: for a given $\mathcal{M}$ and $\cIF_{\mathsf{Adv}}$, how large does $\mathcal{C}$ need to be?   
\paragraph*{\textbf{Connection to Error Detection.}}	
One can note a fairly straightforward relation between tamper detection codes and error detection codes. 
Consider an error-correcting code with minimum distance $d$.
Then, one can use it as a tamper detection code (with $\eps=0$) against an adversary of bounded Hamming weight. 
In tamper detection, we aim to prevent against a much stronger adversary.
 Of course, one can not have an error-detecting code of an arbitrary distance, and hence, tamper detection comes at the cost of some uncertainty in decoding, reflected in Property~\ref{item:intro_soundness_1}. 
 Additionally, in the case of tamper detection, Property~\ref{item:intro_soundness_1} only requires that tampering to be detected. In particular, we have no requirement to recover the original message. In contrast, the Hamming weight bound of $\frac{d}{2}$ on the adversary in the case of error detection guarantees such a recovery. 

\paragraph*{\textbf{Relaxed Tamper Detection.}}
The motivation for a tamper detection code is to construct a storage where it is hard for an adversary to change an encoding of a message to the encoding of some other message. 
A similar effect can be achieved if one considers the following property instead of Property~\ref{item:intro_soundness_1}.

\begin{enumerate}
\item[{\crtcrossreflabel{$\mathrm{B}^\prime$}[item:intro_relaxed_soundness]}]\hspace{-0.25cm}. $\Pr\left( \dec\left( f \left( \enc \left( m \right) \right)\right)= \lbrace m, \perp\rbrace \right) \geq 1-\eps$ (\emph{Relaxed soundness}).
\end{enumerate}
Here, in case of tampering, a decoder is either allowed to detect tampering or output the original message $m$.
In this case, although there was some tampering and the decoder does not necessarily detect it (by outputting $\perp$), the storage still managed to revert back to the original message. Clearly, Property~\ref{item:intro_soundness_1} implies Property~\ref{item:intro_relaxed_soundness}; hence we refer to a code satisfying Property~\ref{item:intro_completeness_1} and \ref{item:intro_relaxed_soundness} as a relaxed tamper detection code.  

\subsection{Previous works}	The above experiment has been extensively studied, both in the weak and the general form, when the message set $\mathcal{M}$ and storage $\mathcal{C}$ are classical strings~\cite{JW15,CG13,FMNV14}.
In the classical setup, one typically has $\mathcal{M}= \lbrace 0,1 \rbrace^k$, $\mathcal{C}= \lbrace 0,1\rbrace^n$, and a tampering family $\mathcal{F}_{\mathsf{Adv}} \subset \mathcal{F}_n \setminus \Id_n$ where $\mathcal{F}_n$ is the set of all possible Boolean functions from $n$-bits to $n$-bits, $\mathcal{F}_n= \lbrace f:\ \lbrace 0,1 \rbrace^n \rightarrow \lbrace 0,1\rbrace^n \rbrace$.
	 Suppose we restrict ourselves to encoding and decoding strategies that are deterministic. In that case, tamper detection schemes do not exist even for the family of additive tampering $\mathcal{F}_{\Delta}=\lbrace f_{\mathrm{e}}(x) = x \oplus \mathrm{e} \rbrace_{\mathrm{e}}$ where $\mathrm{e}\in \lbrace 0,1\rbrace^n \setminus 0^n$. This can be seen as follows: let messages $m_0$ and $m_1$ be any two distinct messages with $\enc(m_0)$ and $\enc(m_1)$ as their corresponding encodings.
  Consider the function $f_{\mathrm{e}}$ for $\mathrm{e} =\enc(m_0) \oplus \enc(m_1)$.
	  The tampering then results in $\dec ( f_{\mathrm{e}}(\enc(m_i)))= m_{1-i}$ for $i \in \{0, 1\}$; making randomness a necessity for tamper detection.  
	  
	  Cramer et al.~\cite{CDFPW08} studied the problem of tamper detection for the function family $\mathcal{F}_\Delta$ and gave corresponding construction of what they refer to as \emph{algebraic manipulation detection codes}.
   \paragraph*{\textbf{Algebraic Manipulation Detection (AMD) Codes.}} These codes provide tamper detection security for the function family $\mathcal{F}_{\Delta}=\lbrace f_{\mathrm{e}}(x) = x \oplus \mathrm{e},\ \mathrm{e} \neq 0\rbrace$.
 Formally,	 
	\begin{fact}[Theorem 2, \cite{CDFPW08}] \label{thm_1}Let $q$ be a prime power and $d$ be a positive integer such that $d<q$.
	 There is an explicit $(\enc,\dec)$ construction that is tamper-secure with parameters $\left(k=d\log q,\ n=(d+2)\log q,\ \eps= \frac{d+1}{q}\right)$ against $\mathcal{F}_{\Delta}$.	\end{fact}

Recall that $k$ and $n$ are bit lengths of message and codewords,
respectively.
The additive overhead of $n$ over $k$ measures the efficiency of AMD codes.
	  An optimal code for parameters $k$ and $\eps$ has the smallest possible $n$.
	    For the security parameter $\eps \leq 2^{-\lambda}$, Fact~\ref{thm_1} gives $(k,\ k+ 2\lambda + 2\log(d+1),\ 2^{-\lambda})$ AMD codes.
	     Thus, the overhead for  codeword length (over the message length) is $2\lambda + 2\log(d+1)$, which was later shown to be optimal up to a multiplicative factor two~\cite{CPX14}.
      
\paragraph*{\textbf{Classical Tamper Detection.}}
	AMD codes provide tamper detection security against a function family of size $2^n$.
	 However, the size of the tampering family $\mathcal{F}_{\mathsf{Adv}}$ can be up to $2^{n 2^n}$ when one considers all classical Boolean functions $f$ from $n$-bits to $n$-bits.
	 Thus, it is interesting to see how big this family can be made, while achieving tamper detection.
	 Again, one can see that it is not possible to construct tamper detection codes for the complete family of size $2^{n 2^{n}}$.
	 For example, consider a family of functions $\mathcal{F}_{\mathrm{con}} = \lbrace f_i(m) = \enc({i}) \rbrace_{i \in \mathcal{M}}$. No scheme can satisfy Property~\ref{item:intro_soundness_1} (or even \ref{item:intro_relaxed_soundness}) for such a family.
  
	Interestingly, Jafargholi and Wichs~\cite{JW15} showed that tamper detection codes indeed exist for any $\cIF_{\mathsf{Adv}}$ of size upto $2^{2^{\alpha n}}$ $(\text{for any constant }\alpha < 1)$,  as long as every function $f \in \cIF_{\mathsf{Adv}}$ satisfies two additional conditions:	
	\begin{itemize}
		\item High min-entropy: $f(U_X)$ has sufficiently high min-entropy\footnote{    For a random variable $X$, it min-entropy is  $\hmin{\left( X \right)} = - \log \left(\max_{x \in \supp(X)} \Pr(X=x) \right)$.}, where $U_X$ is the uniform distribution on the domain of $f$. 
		\item Few fixed points:  There are not too many points such that  $f(x) = x.$
	\end{itemize}  
	The condition of high min-entropy avoids functions that put too much weight on a single point in the output.
	In particular, it avoids functions that are close to constant functions. 
	Similarly, the condition of a few fixed points avoids functions that are close to the identity map. 
	This result shows that tamper detection codes exist against any family that avoids these cases, even for those with size doubly exponential in $n$.
	 Note that this result is based on a probabilistic argument, and as such, it only shows the existence of such codes, and it is not known if they can be constructed efficiently. 
  However, for smaller families (having sizes upto $2^{\mathrm{poly}\left(n\right)}$), one can indeed construct them efficiently~\cite{FMNV14} in the “common reference string” (CRS) model. 

	\subsection{Our results}
	In this work, we aim to extend the scope of the theory of tamper detection to include adversaries that are capable of doing quantum operations.
	 Hence, a family of unitary operators is a natural place to start the discussion. 
  In particular, we consider a setting where the space of codewords $\mathcal{C}$ is of quantum states, and an adversary can apply a unitary operator from a known family of unitary operators $\mathcal{U}_{\mathsf{Adv}}$.      
	 The analogous question of tamper detection can now be asked in different scenarios:
	 \begin{enumerate}[1.]
	 \item \label{question:1} Do tamper detection codes exist when $\mathcal{M}$ is the set of $k$-bit (classical) messages?
	 \item \label{question:2} Do tamper detection codes exist when $\mathcal{M}$ is the set of $k$-qubit (quantum) messages?
	 \item \label{question:3} Can these constructions be made efficient, potentially considering families of relatively small size, say, $\vert \mathcal{U}_{\mathsf{Adv}} \vert \leq 2^{\mathrm{poly}\left(n\right)}$?  
	 \end{enumerate}
	The first and the second question are direct analogues of the tamper detection theory when the adversary has their action defined via a unitary operator (instead of classical bits-to-bits manipulation).
 The first question considers the scenario of protecting classical information from a quantum adversary, whereas, for the second question, the information to be stored is itself quantum.
 The third question is inspired by the fact that efficient classical tamper detection codes (such as AMD codes) exist when the adversarial family has small cardinality. 
We provide affirmative answers to questions 1 and 2 using probabilistic arguments.
Partially addressing question 3, as an example of efficient construction, we show that a natural quantum analogue of classical AMD codes is sufficient for the purpose. 

\paragraph*{\emph{How far does the classical theory take us in question~\ref{question:1}?}}

Before going towards \emph{truly} quantum encoding-decoding strategies, one can ask if the existing classical schemes themselves provide us security against unitary tamperings when $\mathcal{M}$ is classical. 
There is a natural strategy to follow: Consider a classical tamper detection code with the encoder $\mathsf{Enc_{Cl}}$. 
Since the encoder is randomized, for a message $m \in \lbrace 0,1 \rbrace^k$ and randomness $r$, its encoding is given as $\mathsf{Enc_{Cl}}\left( m ,r\right).$ 
Define $\enc_Q(m) =  R_0 \sum_{r} \ket{\mathsf{Enc_{Cl}}\left( m ,r\right)}$, where $R_0$ is an appropriate normalization constant.
After the adversary acts via a unitary $U$, the decoder simply measures in the computational basis, forcing the tampering to be effectively classical. 
Then, one can try to use the classical decoder to recover the message.

The rationale for the above strategy is simple.
Although the tampering can be non-classical (that is, not via a function $f: \lbrace 0,1 \rbrace^n \to \lbrace 0,1 \rbrace^n$), the decoder can first measure in the computational basis.
The resultant operation can now be treated as a (potentially randomized) function from $n$-bits to $n$-bits.
And thus, a unitary adversary followed by the computational measurement can be simulated by a randomized classical adversary given by $\mathcal{F}_{\mathsf{Adv}}$.
However, classical tamper detection does not protect against arbitrary function families.
Thus, one would additionally need a statement of the following form:

Given an adversarial unitary family $\mathcal{U}_{\mathsf{Adv}}$  there exists a classical function family $\mathcal{F}_{\mathsf{Adv}}$ such that:
\begin{enumerate}  \item\label{cond:one_for_simulation} There exists a classical tamper detection code against $\mathcal{F}_{\mathsf{Adv}}$.
    \item \label{cond:two_for_simulation} For every $U \in \mathcal{U}_{\mathsf{Adv}}$ and $c \in \lbrace 0,1 \rbrace^n$, its action followed by measurement in the computational basis can be emulated classically via $\mathcal{F}_{\mathsf{Adv}}$. That is,    
    \[ \Pr\left(\text{measurement results in } c^\prime~  \vert \text{ codeword was $c$}\right) = \langle c^\prime \vert U(c) \vert c^\prime \rangle = \sum_{f \in \mathcal{F}_{\mathsf{Adv}}} P^U(f) \mathbbm{1}_{f(c)=c^\prime},\]    
    where $P^{U}(f)$ is a probability distribution supported on  $\mathcal{F}_{\mathsf{Adv}}$ depending only on $U$ (and not on $c$), and $\mathbbm{1}$ is the indicator function.
\end{enumerate}
Typically, one would need $\vert \mathcal{F}_{\mathsf{Adv}} \vert \leq 2^{2^{\alpha n}}$ for some appropriately chosen constant $\alpha<1$ in addition to every $f$ having enough min-entropy and few fixed points.
Indeed, one can construct such $\mathcal{F}_{\mathsf{Adv}}$ for some families $\mathcal{U}_{\mathsf{Adv}}$.
For example, consider the family of generalized Pauli operators (see Section~\ref{subsec:Generalized_Pauli} for the definition).
\begin{example} \label{example:classical_reduction_for_amd}
The unitary operators in the family are indexed by $a,b$ and given as $\sigma_{a,b}= X_a Z_b$.
A rather straightforward calculation leads to the following:
\begin{itemize}
    \item[-] $X_a Z_b$ acting on any $c$ followed by computational basis measurement  results in $c+a$ with probability $1$.
\end{itemize}
Now, consider $\mathcal{F}_{\Delta}=\lbrace f_{\mathrm{e}}(x) = x \oplus \mathrm{e},\ \mathrm{e} \neq 0\rbrace$.
Define $P^{a,b}({f_e})= \delta_{a,e}$ where $\delta$ is the standard Kronecker delta function.
Then it is easy to verify that for $\sigma_{a,b}$ such that $a \neq 0$, 
\[\langle c^\prime \vert \sigma_{a,b}(c) \vert c^\prime \rangle =  \sum_{e \neq 0 } P^{a,b}(f_e) \mathbbm{1}_{f_e(c)=c^\prime}.\]
\end{example}
Whereas, any $\sigma_{a,b}$ with $a=0$, the codeword is not even perturbed by the action of $\sigma_{a,b}$ as the $Z_b$ operator can only result in adding a global phase to classical messages.
Thus, any (non-trivial) action of a generalized Pauli operator, followed by measurement, can be simulated by $\mathcal{F}_{\mathsf{Adv}}= \mathcal{F}_{\Delta}$.
This gives us the following:
\begin{theorem}[Quantum AMD codes] \label{QAMD_informal}
Let $q$ be a prime power and $d$ be an integer such that $0<d<q$. Let $\mathcal{U}_{\mathsf{P}_N}$ be the group of generalized Pauli operators~\footnote{The generalized Pauli matrices we define are the so-called non-Hermitian Sylvester Generalized Pauli matrices. See Section~\ref{subsec:Generalized_Pauli} for the definition.} acting on $n=\log N$ qubits. 
There exists an efficient $(\enc, \dec)$ scheme that is relaxed tamper-secure against $\mathcal{U}_{\mathsf{P}_N}$ with parameters 
$\left(k=d\log q,\ n=\left( d+2 \right)\log q,\ \eps= \left(\frac{d+1}{q}\right)^2\right)$.
\end{theorem}
Since there exists a family $\mathcal{F}_{\Delta}$ that can simulate generalized Pauli operators, we can directly use the classical scheme to detect a generalized Pauli operator adversary (see Appendix~\ref{qamd} for proof).
However, it is not clear if, for a general unitary family $\mathcal{U}_{\mathsf{Adv}}$ (following some reasonable conditions), there exists a classical family $\mathcal{F}_{\mathsf{Adv}}$ satisfying conditions~\ref{cond:one_for_simulation} and \ref{cond:two_for_simulation}.
And hence, in general, we can not ascertain that the natural quantum analogue of merely taking superpositions of classical encodings will suffice.

Now, we move on to the main contribution of the work, considering  general families of unitary operators.
Recall that classical results are proved under two restrictions.
One, every function has enough min-entropy.
And two, every function has at most a few fixed points (also referred to as the \emph{far from the identity condition}). 
We also provide our results under similar restrictions.
Note that when considering a unitary family, we readily have the min-entropy condition satisfied.
So, we additionally impose a condition that captures closeness to the identity.
We require that for every unitary operator $U\in\mathcal{U}_{\mathsf{Adv}}$, its inner product with the identity map $\left(\left\vert \left\langle \Id, U \right\rangle\right\vert = \left\vert \Tr(U) \right\vert\right)$ is bounded away from $N$.
The main contribution of this work can then be stated as follows:
\begin{theorem}[Quantum tamper detection for quantum messages]\label{qinfo_qt} 
Let $\mathcal{M}$ be the set of quantum messages and
let $\mathcal{U}_{\mathsf{Adv}} \subset \mathcal{U}\left(\mathbb{C}^{2^n}\right)$ be a family of size $2^{2^{\alpha n}}$ for some constant $\alpha < \frac{1}{6}$. 
Moreover, every $U \in \mathcal{U}_{\mathsf{Adv}}$ is such that $\vert \Tr(U) \vert \leq \phi 2^n$, where $\phi$ is a constant strictly less than 1. 
Then there exists a quantum tamper detection code against $\mathcal{U}_{\mathsf{Adv}}$.
\end{theorem} 
Note that in the above theorem, $\phi$ is not an absolute constant but depends on the size of plaintext space $K$ and the security parameter $\eps$.

Although our main motivation is to consider tamperings against quantum messages, as a \emph{warm-up}, we consider the case of classical messages.
This will help us to demonstrate our technique, give a brief overview and establish some bounds that will be used later.

\begin{theorem}[Quantum tamper detection for classical messages]\label{cinfo_qt} Let $\mathcal{M}$ be the set of classical messages and
let $\mathcal{U}_{\mathsf{Adv}} \subset \mathcal{U}\left(\mathbb{C}^{2^n}\right)$ be a family of size $2^{2^{\alpha n}}$ for some constant $\alpha < \frac{1}{6}$. 
Moreover, every $U \in \mathcal{U}_{\mathsf{Adv}}$ is such that $\vert \Tr(U) \vert \leq \phi 2^n$, where $\phi$ is some constant strictly less than 1. 
Then there exists a quantum tamper detection code against $\mathcal{U}_{\mathsf{Adv}}$.
\end{theorem}	

We also show that even if one drops the condition on trace, we can achieve a \emph{relaxed} version of quantum tamper detection (where quantum counterparts of Property~\ref{item:intro_completeness_1} and \ref{item:intro_relaxed_soundness} are satisfied).
Again, here we state the theorem informally.
The formal statement, along with its proof, is presented in Section~\ref{Sec:Relaxed_TD}.

\begin{theorem}[Relaxed quantum tamper detection for classical messages]\label{cinfo_rqt} 
 Let $\mathcal{M}$ be the set of classical messages and
let $\mathcal{U}_{\mathsf{Adv}} \subset \mathcal{U}\left(\mathbb{C}^{2^n}\right)$ be a family of size $2^{2^{\alpha n}}$ for some constant $\alpha < \frac{1}{6}$. 
Then there exists a relaxed quantum tamper detection code against $\mathcal{U}_{\mathsf{Adv}}$.
\end{theorem}	 
Note that, Theorem~\ref{cinfo_rqt} allows us to also include operators that are close to the identity operator.
It is not hard to see that such a relaxation to Property~\ref{item:intro_relaxed_soundness} is necessary as one can not satisfy Property~\ref{item:intro_soundness_1} with such operators.

\paragraph*{\textbf{Proof overview.}}	 
	 Similar to the proof provided by~\cite{JW15}, our proofs for Theorem~\ref{qinfo_qt},~\ref{cinfo_qt} and~\ref{cinfo_rqt} use probabilistic arguments via Chernoff-like tail bounds for limited independence.
Before going ahead, we would like to fix some notation.

For a matrix $A$, let $A(i)$ denote the $i$-th column of $A$, which we will often treat as a vector.
When dealing with classical messages, we will denote them as $m \in \mathcal{M}$, whereas for quantum messages, we will use $\ket{m} \in \mathcal{M}$ (or $\ket{s} \in \mathcal{M}$ to explicitly indicate that the message is in superposition).
Moreover, we use $K=2^k$ and $N=2^n$, for ease of presentation.

  Let us first consider the case when $\mathcal{M}$ is the set computational basis states, $\mathcal{M}=\lbrace \ket{m}, m \in \lbrace 0,1 \rbrace^k \rbrace$.
  Our scheme uses a strategy where encoding is done by a Haar-random isometry $V$.
 For a fixed $V\in \mathcal{U}\left(\mathbb{C}^{N}\right)$, our encoding scheme is fairly natural; we encode a classical message $m$ as the $m$-th column of $V$, giving $\enc(m)= \ket{V(m)}$.
  
  Then the quantum tampering experiment can be thought of as below: 
\begin{enumerate}[1.]
\item A $k$-bit message $m$ is encoded in $n$-qubits via $V$, resulting in $\vert \psi_m \rangle= \enc(m) = \ket{V(m)}$.
\item An adversary then tampers with $U \in\mathcal{U}_{\mathsf{Adv}}$, resulting in the state $U \vert \psi_m \rangle \langle \psi_m\vert U^\dagger$.
\item\label{measurement_in_3} For $i \in \lbrace 1,2,\ldots,K\rbrace$, define $\Pi_{i} =\ket{\psi_i}\bra{\psi_i}$ and let $\Pi_{\perp}= \Id- \sum_{i} \Pi_{i}$. 
The decoder measures with the POVM $\lbrace  \Pi_1 ,  \ldots, \Pi_{K},\Pi_{\perp}\rbrace$.
\item If the measurement results in $\Pi_\perp$, then abort with detection of tampering. 
Otherwise, apply ${V}^\dag$ and output the resulting candidate message $\hat{m}$.
\end{enumerate} 	 

The completeness of the protocol is easy to check.
To show that the above encoding-decoding is $\eps$-tamper secure, one needs $\Pi_{\perp}$ to be a high probability event for any non-trivial tampering;  $\Tr\left(\Pi_{\perp} U \vert\psi_m\rangle \langle \psi_m \vert U^{\dag}\right) \geq 1-\eps$ (\emph{Soundness}). 
For that, we define the following random variables:
\begin{enumerate}
\item[•] $X_{js}= \vert \langle \psi_j \vert U \vert \psi_s\rangle \vert^2$ denotes the probability that message $s$ was decoded to $j$. 
\item[•] $X_{s}= \sum\limits_{j\neq s} X_{js}$ denotes the probability of decoding $s$ to a message other than $s$ and $\perp$.
\end{enumerate}
 Measurement results in either the same $\Pi_s= \ket{\psi_s}\bra{\psi_s}$ (with probability $P_\mathsf{same} = X_{ss}$) or one of the $\Pi_j = \ket{\psi_j}\bra{\psi_j}$ that is different from $s$ (with probability $P_{\mathsf{diff}}$) or $\Pi_\perp$ that indicates the tampering (with probability $P_{\perp}$). Thus, $P_{\mathsf{same}}+P_{\mathsf{diff}}+ P_{\perp}=1$.
Recall that we need to lower bound the probability of obtaining $ P_{\perp}$.
We do this by upper bounding $P_\mathsf{same}$ and $P_{\mathsf{diff}}$, which requires us to prove sharp Chernoff-like tail bounds for random variables $X_{ss}$ and $X_{s}$, respectively. This completes our proof for $\mathcal{M}=\lbrace 0,1 \rbrace^k$.

The setup when $\mathcal{M}$ is quantum (that is, messages to be stored are $k$-qubit states), is slightly more involved. 
Let $\ket{s} \in \mathcal{M}$ be a message that we want to store.
 Note that we need to preserve not only $2^k$ basis states but also the arbitrary superposition; arbitrary message $\ket{s}$ is a linear combination of computational basis states $\ket{s}= \sum_{i} \alpha_i \ket{b_i}$.
 Suppose one uses a direct linear extension of the earlier encoding-decoding strategy, $\enc(\ket{s})= \sum_i \alpha_i \enc{(b_i)}$.
 The measurement in step~\ref{measurement_in_3} is done over the basis encodings $\lbrace \enc(b_i) \rbrace$, and hence it can destroy the superposition.  
 To recover $\ket{s}$, it is necessary to keep $\ket{s}$ intact, and in particular, the resulting state after the measurement should not be disturbed too much from the pre-measurement state $\enc(\ket{s})$.
 To remedy this, we modify the decoder slightly, where we do measurement with a two-outcome POVM (instead of $K+1$ outcomes).
The binary POVM we use corresponds to the projection on $\enc(\mathcal{M}) = V(\mathcal{M})$ (and its orthogonal complement).
Hence, for $\ket{s} \in \mathcal{M}$, we require that any adversarial unitary $U \in \mathcal{U}_{\mathsf{Adv}}$ takes $\enc(\ket{s})$ to a vector in the orthogonal complement of $V(\mathcal{M})$.
This reduces the problem of tamper security of $\ket{s}$  to Chernoff-like tail bounds for a slightly different random variable $X_m= \bra{\psi_m} U\  (\enc\ket{s})\left(\enc\ket{s}\right)^{\dagger} \ U^\dagger \ket{\psi_m}$.

To prove sharp Chernoff-like tail bounds for random variables $X_{ss}$, $X_{s}$, and $X_m$, we use techniques from representation theory.
The proof uses Weingarten calculus and some properties of the symmetric group.

 We note that Theorem~\ref{cinfo_qt} (regarding the security of classical messages) follows as a corollary of Theorem~\ref{qinfo_qt} (regarding the security of quantum messages), as the former is a strict subset of the latter.
 Nonetheless, we include it as we also show Theorem~\ref{cinfo_qt} in the relaxed form, on an adversarial family with no trace bound needed (see Theorem~\ref{cinfo_rqt}).
 This is further used to show the existence of \emph{non-malleable codes} via a standard reduction (see Theorem~\ref{thm:non-malleability}) against a unitary family of size upto of size $2^{2^{ n/6}}$.

\subsection*{Related Works and Future Directions.}

Since Shor's work on the existence of error-correcting structures for the quantum framework~\cite{Shor_coding}, there has been a rich history of quantum error correction~\cite{CSS,Dan_thesis, Kitaev,steane}. 
One can draw similar parallels between quantum error correction and quantum tamper detection as those present in the classical framework. 
In particular, tamper detection schemes try to handle an error set that is not bounded by weight with a possible loss in the ability to correct.

\textbf{Quantum Authentication Schemes (QAS).} The work of~\cite{GC17,ABW09} studies the notion of non-malleability in quantum authentication schemes. In quantum authentication schemes, both the encoder and decoder have a pre-shared private random key $\mathrm{K}$ that is not accessible to an adversary.
We require that in the absence of an adversary, the received state should be the same as the sent state, and otherwise, with high probability, either the decoder rejects, or the received state is the same as
that sent by the encoder. It is known that such quantum authentication schemes exist (for example, Clifford authentication~\cite{GC17}), whereas tamper detection schemes are keyless.
Similarly, a few other works have also considered a ``tampering" adversary~\cite{Anne,Dan_clone}.
Again, these works are keyed primitives, making them different from tamper detection that works without keys.

Classically, tamper detection codes have turned out to be a fruitful object with rich applications. The work of~\cite{DPW10} introduced non-malleable codes for which decoding a tampered codeword either results in an original message or a message unrelated to $m$.
The work of~\cite{JW15} made the connection between tamper detection and non-malleable codes more explicit; by giving a modular construction of non-malleable codes out of weak tamper detection codes and leakage-resilient codes.
There is a vast body of literature that considers tampering attacks using other approaches besides non-malleable codes and tamper detection codes (see~\cite{MDRHX11,MDT10,MTE03,MKSXK12,SKD11,RATST04,VAVY11,YMAD06,TBA11,KR12}).
We refer to~\cite{DPW10} for a more detailed comparison between these approaches and non-malleable codes, which have been a central object of study in recent times.

\subsection{Subsequent works on tamper detection and non-malleable codes\label{otherworslabel}}

\subsubsection*{On tamper detection in the qubit-wise tampering model}

 In \cite{Ber23}, Bergamaschi studied a particular subclass of tamper detection codes, namely, against an adversary holding only Pauli operators.
 In what they refer to as \emph{PMD codes}, they construct an efficient tamper detection scheme against such a Pauli adversary when (plaintext) messages are quantum.
 Hence, as mentioned by them, PMD codes can be thought of as a natural generalization of quantum AMD codes.
 We would also like to point out that the existence of such codes for quantum messages is also implied by our work as the family of Pauli operators falls within the scope of Theorem~\ref{theorem:td_quantum_messages}.
 As an application, they use PMD codes to construct keyless authentication codes against qubit-wise tamperings, a task that is provably impossible, solely with a classical encoding. 
 
\subsubsection*{On non-malleability in the split-state tampering model}
 In another work, Aggarwal, Boddu and Jain~\cite{ABJ22} defined the notion of non-malleable codes for classical messages against quantum adversaries (having access to shared entanglement) in the \emph{split-state model}, where cipher-text is split into two parts, and the adversary is allowed to tamper them independently (via unitaries of the form $U_1 \otimes U_2$).

\begin{figure}[ht]
\centering
\begin{tikzpicture}

\node at (1,4.5) {$m$};
\draw (1.2,4.5) -- (5,4.5);
\draw (5,4.5) -- (11,4.5);
\node at (11.5,4.7) {$m$};

\draw (2,1.5) -- (3,1.5);
\draw (3,-0.5) rectangle (4.5,3.5);
\node at (3.8,1.5) {$\enc$};

\draw  (2,1.5) -- (2,4.5);


\draw (4.5,1.6) -- (5.2,1.6);
\draw (7.8,1.6) -- (9,1.6);
\node at (11.5,1.6) {$m'$};
\draw (11,1.4) -- (12,1.4);



\node at (6.5,1.4) {$\phi \in \Phi_{\mathsf{Adv}}$};
\draw (5.2,-0.8) rectangle (7.8,3.8);



\draw (9,-0.5) rectangle (11,3.2);
\node at (10,1.5) {$\dec$};

\end{tikzpicture}
\caption{Tampering process.}\label{fig:splitstate1}
\end{figure}

\begin{definition}[\cite{ABJ22}~non-malleable codes against adversary family~${\Phi}_{\mathsf{Adv}}$]\label{def:nmcodes}
We say that an encoding-decoding scheme $\left(\enc, \dec \right)$ (see Definition~\ref{def:enc-dec_schemes} and Figure~\ref{fig:splitstate1}) is $\eps$-non-malleable secure against adversary family~${\Phi}_{\mathsf{Adv}}$ for classical messages $\mathcal{M}$, if for all $m\in \mathcal{M}, \phi \in {\Phi}_{\mathsf{Adv}}$, the following holds:
\[\dec \left( \phi \left( \enc(m)  \enc(m)^\dagger \right)   \right) \approx_{\eps} p_\phi m + (1-p_\phi) \eta_{\phi},\]
where $(p_\phi, \eta_{\phi} )$ depend only on adversary $\phi$. Here $p_\phi \in [0,1]$ and $\eta_{\phi}$ are independent of original message $m$.
\end{definition}
This work considers a much more general class of unitaries (which are not necessarily in a split form).
Of course, this comes at the cost that their constructions are explicit and efficient, whereas our constructions are probabilistic and existential.
Note that this is also seen in the classical tamper detection literature, where split-state codes are efficient, whereas the codes against a general adversary are known to exist (without any explicitly known construction). 
	\section*{Organization of the paper}
	 
	For a quantum adversary with access to unitary operators, the \emph{Haar} measure is the canonical measure to work with.
 For getting bounds on unitary operators, we use Weingarten functions as a tool.
	  Well-known, relevant results are summarized in Section~\ref{wuc}. Additionally, Section~\ref{sec2} also contains elementary observations on permutation groups, along with some technical proofs.
	  In Appendix~\ref{qamd}, we prove Theorem~\ref{QAMD_informal}; in Section~\ref{cqtdc}, we prove Theorem~\ref{cinfo_qt} and Theorem~\ref{cinfo_rqt}; and in Section~\ref{qtdc}, we prove Theorem~\ref{qinfo_qt}.
	  All the proofs involve technical tail bounds regarding moments of certain random variables, which we include in  Appendix~\ref{sec:app_classical} and \ref{sec:app_quantum}.

\section{Preliminaries}\label{sec2}
\subsection{Some notation}
All the logarithms are evaluated to the base $2$. Consider a finite-dimensional Hilbert space $\cH$ endowed with an inner product $\langle \cdot, \cdot \rangle$ (we only consider finite-dimensional Hilbert spaces).
For $p \geq 1$ we write $\| \cdot \|_p$ for the Schatten $p$-norm. 
We use $\rho_1 \approx_{\eps} \rho_2$ to mean that $\Vert \rho_1 - \rho_2 \Vert_1 \leq \epsilon$.
A similar convention will be followed for two probability distributions as well.
A quantum state (or a density matrix or a state) is a positive semi-definite matrix on $\cH$ with the trace equal to $1$. It is called {\em pure} if and only if its rank is $1$.  Let $\ket{\psi}$ be a unit vector on $\cH$, that is $\langle \psi,\psi \rangle=1$.
The topological space of norm-1 vectors (the unit $N$-sphere) in a normed $N$-dimensional vector space $V$, is denoted as $\mathbb{S}^{(N-1)}(V)$. 
When $V$ is clear from the context, we drop it.
For an $n$-dimensional vector $v$, we will use the standard notation $v=(v_1, \ldots, v_n)$ and thus $v_i$ will refer to the $i$-th coordinate.
Similarly, for a matrix $M$, we will denote its $i,j$-th entry by $M_{ij}$.

\begin{itemize}
    \item[-] A {\em unitary} operator $U:\cH \rightarrow \cH$ is such that $U^{\dagger}U = U U^{\dagger} = \id$.
The set of all unitary operators on $\cH$ is  denoted by $\mathcal{U}(\cH)$.
\item[-] An {\em isometry}  $V:\cH_A \rightarrow \cH_B$ is such that $V^{\dagger}V = \id_A$ and $VV^{\dagger} = \Pi_{V(\mathcal{H}_A)}$, where $\Pi_{V(\mathcal{H}_A)}$ is the projection on the image of $\mathcal{H}_A$ under $V$.
\item[-] A {\em POVM}  $\left\lbrace M, \id - M \right\rbrace$ is  a $2$-outcome quantum measurement for $0 \le M \le \id$.
We use the shorthand $\overline{M} = \id - M$, where $\id$ is clear from the context. Similarly, a measurement $M_A$ acting on a combined space $\mathcal{H}_A \otimes \mathcal{H}_B$ will be used to represent $M_A \otimes \id_B$.
\item[-] A $k$-outcome POVM is defined by a collection $\lbrace M_1, M_2, \ldots, M_k \rbrace$, where $0 \leq M_i \leq \id$ for every $i \in [k]$ and $\sum_{i} M_i=\mathbb{I}$. 
\end{itemize}
We now state the following useful facts. 
\subsection{Some elementary bounds}
\begin{fact}\label{ub_on_factorial} For any integer $n \geq 1$ 
	
	\[ \frac{n^n}{e^{n-1}} \leq n! \leq \frac{n^{n+1}}{e^{n-1}}.   \]
\end{fact}
 
\begin{fact} \label{fact:unitary_trace_leq_n}
Let $U$ be a unitary operator on $\mathbb{C}^N$. Then $\vert \Tr(U) \vert \leq N$.
\end{fact}

\begin{fact}\label{ub_on_ncr} For positive integers $k,n$ such that $1 \leq k \leq n$, 	
	\[ \left(\frac{n}{k}\right)^k \leq \binom{n}{k} \leq \left(\frac{en}{k}\right)^k.   \]
\end{fact}

\subsection{Definitions for Tamper Detection}
\begin{definition}[$\eps$-net (Lemma 5.2, \cite{V10})]\label{epsnets}
Fix an $\eps >0$. Then there exists an integer $N$ and a set of vectors $\{ \ket{\psi_1},  \ket{\psi_2} , \ldots ,  \ket{\psi_N}\}$ in $\mathbb{S}^{d-1}$ such that the following properties hold:
\begin{itemize}
    \item $N \leq \left(\frac{4d}{\eps}\right)^d$. 
    \item For any state $\ket{\psi} \in \mathbb{S}^{d-1}$, there exists $j \in [N]$ such that $\Vert \ket{\psi} -  \ket{\psi_j} \Vert_1 \leq \eps.$ 
    \end{itemize}       
\end{definition}
\begin{definition}[Quantum encoding and decoding schemes\label{def:enc-dec_schemes}]
Let $\enc: \mathcal{M} \longrightarrow \mathbb{S}^{N-1}$ be a map and $\dec: \mathbb{S}^{N-1} \longrightarrow \mathcal{M} \cup \lbrace \perp \rbrace$.
Then, we say that $(\enc,\dec)$ is an encoding-decoding  scheme if the following holds:
for all $m \in \mathcal{M}$, $\Pr \left( \dec \left( \enc(m) \right) =m \right)=1$.
\end{definition}\begin{definition}[Tamper detection (against unitary adversaries)]
Let $\mathcal{U}_{\mathsf{Adv}} \subset \mathcal{U}\left( \mathbb{C}^N\right)$ be a family of unitary operators. 
We say that an encoding-decoding scheme $\left(\enc, \dec \right)$ is $\eps$-tamper secure against family $\mathcal{U}_{\mathsf{Adv}}$ for messages $\mathcal{M}$, if for all $m\in \mathcal{M}, U \in \mathcal{U}_{\mathsf{Adv}}$, the following holds:
\[\Pr\left(\dec \left( U \left( \enc(m) \right) \left( \enc(m) \right)^\dagger U^\dagger \right)= \perp \right) \geq 1-\eps.\]
Furthermore, if $\mathcal{M} = \lbrace 0,1 \rbrace^{k}$, we say that $(\enc,\dec)$ is $\left(K=2^k,N,\eps\right)$-tamper secure for classical messages, whereas if $\mathcal{M}= \mathbb{S}^{K-1}$, we say that $(\enc,\dec)$ is $\left(K, N, \eps \right)$-tamper secure for quantum messages.
\end{definition}
Now we define a \emph{relaxed version} of tamper detection. In this version, the aim of a decoder is to either detect tampering or output the original message.
Compared to the original definition of tampered detection, the relaxed version has a scope to revert a tampering, without even detecting it.
Since our result holds for classical messages (against unitary tamperings), we define relaxed tamper detection only for classical messages but one can define an analogous notion for quantum messages as well.
\begin{definition}[Relaxed tamper detection\label{def:relaxed_TD}]
Let $\mathcal{U}_{\mathsf{Adv}} \subset \mathcal{U}\left( \mathbb{C}^N\right)$ be a family of unitary operators and let $\mathcal{M}= \lbrace 0,1\rbrace^k$. 
We say that an encoding-decoding scheme $\left(\enc, \dec\right)$ is $(K,N,\eps)$-tamper secure in the relaxed setting (against  $\mathcal{U}_{\mathsf{Adv}}$), if for all $m \in \mathcal{M},\ U\in \mathcal{U}_{\mathsf{Adv}}$, the following holds:
\[\Pr\left(\dec \left( U \left( \enc(m) \right)\left( \enc(m) \right)^\dagger U^\dagger \right)= \lbrace \perp,m \rbrace \right) \geq 1-\eps.\] \end{definition}
\begin{definition}[Adversarial unitary families]
Let $\mathcal{U}_{\mathsf{Adv}} \subset \mathcal{U}\left( \mathbb{C}^N\right)$ be a family of unitary operators such that the following holds:
\begin{itemize}
    \item[1.] For all $U \in \mathcal{U}_{\mathsf{Adv}}$, we have, $\vert \mathrm{Tr} \left(U\right) \vert \leq \phi N$.
    \item[2.] $\vert\mathcal{U}_{\mathsf{Adv}}\vert \leq 2^{N^{\alpha}}$.
\end{itemize}
We call $\mathcal{U}_{\mathsf{Adv}}$ as an $\left(N, \alpha, \phi \right)$ adversarial unitary family or simply $\left(N, \alpha, \phi \right)$ family. 
\end{definition}


\begin{definition}[Random Haar encoding and decoding schemes\label{def:random_haar_enc_dec}]
		Let $\mathrm{H}$ be a random unitary drawn from $\mathcal{U}\left(\mathbb{C}^N\right)$ (according to the Haar measure).
		Let $V$ be the following matrix constructed by restricting $\mathrm{H}$ to its first $K$ columns: 
	\[V = (\mathrm{H}_1, \mathrm{H}_2,  \ldots, \mathrm{H}_{K}).\]
	
	 Note that $V$ is an isometry.
	 Consider the following encoding and decoding scheme.
	 \begin{itemize}
	     \item[-] Let $V(i)$ denote the $i$-th column of $V$.   
      For $m \in \left[ K\right]$, define $\enc(m)= \vert \psi_m \rangle = \ket{ V(m)}$.
      
      If the message set $\mathcal{S}$ is quantum, the extension is canonical.
      For $\ket{s}= \sum_m \alpha_m \ket{ m}$, the encoding is $\enc(\ket{s})= \sum_m \alpha_m \enc( {m})$.
	     \item[-] $\dec$ to be implemented according to the following procedure:
	     
	     Let $\Pi_{i} = \vert \psi_i \rangle \langle \psi_i \vert$ and $\Pi_{\perp}= \id - \sum_{i} \Pi_{i}$.
	     To decode a message $\vert \theta \rangle$, we measure $\vert \theta \rangle$ in a two-valued POVM $\big\lbrace \sum_{i} \Pi_i, \Pi_\perp  \big\rbrace$.
	     Let $ \psi^\prime $ be the post-measurement state. 
	     If the measurement results in $\perp$ then abort \textrm{\emph{(}indicating tamper detection\emph{)}}; otherwise the decoder outputs $V^\dagger(\psi^\prime)V$.   
	 \end{itemize}
	 Note that if the message set is classical \emph{\big{(}}$\mathcal{M}= \lbrace 0,1 \rbrace^k$\emph{\big{)}}, then the decoder can be reduced to the following action:
	 \begin{itemize}
	     \item[-] $\dec_\mathsf{Cl}$: Measure $\vert \theta\rangle$ in the POVM $\lbrace \Pi_{1}, \Pi_{2}, \ldots, \Pi_{K}, \Pi_{\perp}\rbrace$ , if it results in $\Pi_i$ then output $i$.
	 \end{itemize}
\end{definition}
Below we give the necessary details of permutation groups, generalized Pauli matrices, Haar random unitary operators, and Weingarten unitary calculus, which will be required to state our results. We refer the reader to~\cite{Gu13} for details on Weingarten unitary calculus.
\subsection{Permutation groups}\label{permu_group}

Let $S_n$ be the symmetric group of degree $n$ acting canonically on the set $\left[n\right]:=\lbrace 1,2,\ldots, n\rbrace$. Let $\mathrm{H} \leq S_n$ be a permutation group. For $x \in \left[ n \right]$, \textit{orbit} of $x$ under $\mathrm{H}$, denoted as $\mathcal{O}_\mathrm{H}(x)$ is the set of elements that can be reached from $x$ via $\mathrm{H}$, 
\[\mathcal{O}_\mathrm{H}(x)= \lbrace y\in [n]:\  \exists h\in \mathrm{H},\ y=h(x)\rbrace.\]

We say that $x$ is fixed by $\mathrm{H}$ if $\mathcal{O}_{\mathrm{H}}= \lbrace x\rbrace$. Otherwise, we say that $\mathrm{H}$ moves $x$. We denote the set of elements fixed by $\mathrm{H}$ as $\mathsf{Fix}\left( \mathrm{H}\right)$ and the set of elements moved as $\mathsf{Move}\left( \mathrm{H}\right)$. By extension, for $\sigma \in S_n$ we write $\mathsf{Fix}(\sigma)$ and $\mathsf{Move}(\sigma)$ to mean $\mathsf{Fix}\left(\langle \sigma \rangle\right)$ and $\mathsf{Move}\left(\langle \sigma \rangle\right)$ respectively, where $\langle \sigma \rangle$ is the group generated by $\sigma$.

Given a $\sigma\in S_n$, orbits for $\mathrm{H}=\langle\sigma\rangle$ partition the set $\left[ n\right]$ into disjoint subsets as $\mathcal{O}_{\mathrm{H}}$ gives an equivalence relation. When one writes $\sigma$ as a permutation in a disjoint cycle form, each orbit is a cycle of $\sigma$ and each cycle is an orbit, and hence, we denote an orbit (or a disjoint cycle) by $c$.
Let $C\left(\sigma\right)$ denote the set of orbits $c$ under $\mathrm{H}=\langle\sigma\rangle$.  

For an orbit $c$, let $odd(c)$ denote the number of odd elements in it and $even(c)$ be the number of even elements in it. We define an evaluation map $\mathsf{Val}$ on orbits of $\sigma$. An orbit $c$ is given a value equal to the difference between the number of odd and even elements it contains.
\[\mathsf{Val}(c)= \vert odd(c)- even(c) \vert. \] We also extend the evaluation map to $S_n$ by assigning a value for each permutation. In this case, a permutation will get a value equal to the sum of the values of all of its orbits.

\[\mathsf{Val}(\sigma)= \sum_{c\in C(\sigma)} \mathsf{Val}(c)=  \sum_{c\in C(\sigma)} \vert odd(c)- even(c) \vert.\]

We denote the set of orbits with value $1$ by $C_1(\sigma)$.
It is easy to see that $\sigma$ has full valuation $n$ if and only if it preserves the parity; that is, it takes odd elements to odd elements and even elements to even elements. 

A transposition is a cycle of size $2$.
Every permutation $\sigma \in S_n$ can be written as a product of transpositions.
Let $\mathsf{T}(\sigma)$ denote the minimum number of transpositions required to obtain $\sigma$.
It is known that $\mathsf{T}(\sigma)+ \vert C(\sigma) \vert=n$.
We use $e$ to represent identity permutation. 

\begin{lemma}\label{lemma_fix} For any $\sigma \in S_n$, we have $\vert \mathsf{Fix}(\sigma)\vert \geq 2 \vert C(\sigma)\vert -n$.
\end{lemma}
\begin{proof}
Clearly if $\vert C(\sigma)\vert \leq \frac{n}{2}$, the lemma is trivially true. Suppose $\vert C(\sigma)\vert >\frac{n}{2}$. 

 Note that $\mathsf{Move}(\sigma) \leq 2 \mathsf{T}(\sigma)$ and hence $\left\vert \mathsf{Move}(\sigma) \right\vert \leq 2(n- \vert  C(\sigma)\vert)$. Since $\mathsf{Move}(\sigma)+\mathsf{Fix}(\sigma)=n$, we get $\vert\mathsf{Fix}(\sigma) \vert \geq 2 \vert C(\sigma)\vert -n$.
\end{proof}

\begin{observation}\label{obs_multiplication_cycle} It follows from the definition that, for any permutation $\sigma$ and for any transposition $\tau$, $\mathsf{T}(\tau \sigma)= \mathsf{T}(\sigma \tau) \leq \mathsf{T}\left(\sigma\right)+1$. Also, the number of cycles can increase or decrease by $1$. If elements moved by $\tau$ are in the same cycle of $\sigma$, then $C$ increases by $1$, and if they are in different cycles it decreases by 1.  $\vert C(\sigma)\vert-1 \leq \vert C(\sigma \tau) \vert = \vert C(\tau \sigma) \vert \leq \vert C(\sigma)\vert+1$. 
\end{observation}
For $i \in [0:n-1]$, let $\Sigma_i := \lbrace \sigma \in S_n :  \mathsf{T}(\sigma) =i\rbrace $ denote the  number of permutations $\sigma$ such that the number of transpositions in $\sigma$ is $i$.

\begin{observation}\label{upper_bound_on_Sigma} For $i \in [0:n-1]$ we have, $\vert \Sigma_i \vert \leq  {\binom{n}{2}}^i$.
\end{observation}

Let $\mathcal{B}_{2n}$ be the set of permutations on $2n$ letters that take odd elements to even elements and vice-versa. 
\[\mathcal{B}_{2n}:= \lbrace \beta \in S_{2n} : \text{ for all }x, x +  \beta(x)= 1 \mod 2 \rbrace.\] 	

\begin{lemma} For any $\alpha\in S_{2t}$ and $\beta \in \mathcal{B}_{2t}$, we have $\vert \mathsf{C}(\beta \alpha) \vert -  \mathsf{T}(\alpha) \leq t$.		
\end{lemma}
\begin{proof}
	We will prove this by induction on $\mathsf{T}(\alpha)$. 
	
	\textsf{Base Case:} $ \mathsf{T}(\alpha)=0$, that is, $\alpha=\emph{e}$.  Note that for $\beta \in \mathcal{B}_{2n}$, every cycle must have a length of at least two as $\beta$ can not fix any element. Thus, $\vert C(\beta)\vert \leq \frac{2t}{2} = t$. 
	
	\textsf{Induction Hypothesis (IH):} For all $\alpha^{\prime}$ such that $\mathsf{T}(\alpha^\prime)\leq T_{0}-1$, we have, $\vert \mathsf{C}(\beta \alpha^\prime) \vert -   \mathsf{T}(\alpha^\prime) \leq t$. 
	
	We will show that the upper bound holds for $\alpha$ with $\mathsf{T}(\alpha)=T_0$.
 
	\textsf{The General Case:} Let $C(\alpha)= \lbrace C_1, C_2, \ldots, C_l\rbrace$.
 Since $\alpha \neq \emph{e}$, there exists a cycle of length strictly greater than one.
 Without loss of generality, let that be $C_l$ and $C_l= \left(x_1\ x_2\ \ldots\ x_m \right)$.
 
 Let $\alpha^\prime= C_1 C_2 \ldots C_{l-1} \left(x_1 \ x_2 \ \ldots \ x_{m-1} \right) \left( x_m\right)$. Alternatively,  $\alpha^\prime$ can be obtained from $\alpha$ by fixing $x_m$, that is,
\begin{align*}
	\alpha^\prime(x) &= x_1  &  \hspace{2cm} \text{ if } x = x_{m-1} \\
	 &= x_m   &  \hspace{2cm} \text{ if } x = x_m\\
     &= \alpha(x) &  \hspace{2cm} \text{ if } x\notin \lbrace x_m,x_{m-1} \rbrace.
\end{align*}

	$\vert C(\alpha^\prime)\vert =\vert C(\alpha) \vert+1$ which gives $\mathsf{T}(\alpha^\prime)= \mathsf{T}(\alpha)-1$. By \textsf{IH}, $\vert \mathsf{C}(\beta \alpha^\prime) \vert -  T(\alpha^\prime) \leq t$.  Also, $\alpha= \alpha^\prime \left( x_{m-1} \ x_m \right)$. Thus, $\vert C(\beta \alpha)\vert= \vert C(\beta \alpha^\prime \left( x_{m-1} \ x_m\right)\vert \leq C(\beta \alpha^\prime) +1$. The inequality follows from Observation~\ref{obs_multiplication_cycle}. Putting this along with $\mathsf{T}(\alpha)= \mathsf{T}(\alpha^\prime)+1$ we get the lemma.   
\end{proof}
Since $\mathsf{T}$ is invariant under inverse, we can replace $\mathsf{T}(\alpha)$ by $\mathsf{T}(\alpha^{-1})$. Furthermore, $\mathsf{T}(\alpha)+ \vert C(\alpha)\vert= 2t$. Hence we get the following:  
\begin{cor}\label{upper_bound_on_C} For $\alpha\in S_{2t}$ and $ \beta \in \mathcal{B}_{2t}$, we have, $\vert C(\alpha) \vert + \vert C(\beta \alpha^{-1}) \vert \leq 3t$.
\end{cor}	
\subsubsection{Generalized Pauli matrices} \label{subsec:Generalized_Pauli}
Let $q$ be a prime power and $\mathbb{F}_q$ be the field of size $q$.
And let $\omega$ denote the $q$-th primitive root of unity.
Let $X_a$ and $Z_b$ be the following collection of  operators indexed by $a,b \in \mathbb{F}_q$. 
\begin{align*}
    X_a= \sum_{x  \in \mathbb{F}_q}\ket{x+a}\bra{ x} \\
Z_b= \sum_{x  \in \mathbb{F}_q}\omega^{bx} \ket{x}\bra{ x}.
\end{align*}
The group of generalized Pauli matrices is generated by $\left\langle X_1, Z_1 \right\rangle$.
Generalized Pauli matrices obey the twisted commutation relations given by \[ \quad  X_a Z_b = \omega^{-ab} Z_b  X_a. \]
\subsection{Weingarten unitary calculus}\label{wuc} 
Weingarten functions are used for evaluating integrals over the unitary group  $\mathcal{U}(\bbC^N)$ of products of matrix coefficients~\cite{Gu13}.
 The expectation of products of entries (also called moments or matrix integrals) of Haar-distributed unitary random matrices can be described in terms of a special function on the permutation group. Such considerations go back to Weingarten~\cite{Wg78}, Collins~\cite{C03}. This function is  known as the (unitary) Weingarten function and is denoted by  $\wg$.
 Let $S_p$ be the symmetric group on $[p] =\{1, 2, . . . , p\}$. Let 
 $i = (i_1, \ldots, i_p),\ i' = (i'_1, \ldots, i'_p)$ be $p$-tuples of positive integers from $\{1,2,\ldots, N\}$. We use the notation $\delta_{\sigma}(i,i') = \delta_{i_1i'_{\sigma(1)}} \delta_{i_2i'_{\sigma(2)}} \ldots \delta_{i_pi'_{\sigma(p)}}$, where $\delta$ is the standard Kronecker delta function. For $\alpha \in S_p,  p \leq N$,
\[ \wg(\alpha,N) = \int_{\mathcal{U}(\mathbb{C}^N)} \! U_{11} \ldots U_{pp} \overline{U_{1\alpha(1)}} \ldots \overline{U_{p\alpha(p)}} dU   \]
where ${U}$ is a Haar-distributed unitary random matrix on $\mathbb{C}^N$, $dU$ is the normalized Haar
measure, and $\wg$ is called the (unitary) Weingarten function.
	A crucial property of the Weingarten function $\wg(\alpha,N)$ is that it depends only on the conjugacy class (or alternatively, on the cycle structure) of permutation $\alpha$. 
 So, $\wg\left( \alpha,N \right)$ can as well be denoted as $\wg\left(\left[ l_1, l_2, \ldots,l_{\vert C(\alpha)\vert} \right], N  \right)$ where $c_1, c_2,\ldots,c_{ \vert C(\alpha)\vert}$ are cycles of $\alpha$ having lengths $l_1,l_2, \ldots, l_{\vert { C(\alpha)}\vert}$ respectively.
\begin{fact}[General matrix integration~\cite{CB06,Wg78,C03}]
\label{wgt}		

	Let $N$ be a positive integer and $i = (i_1, \ldots, i_p), ~ i'=(i'_1, \ldots, i'_p), \ j = (j_1 ,\ldots, j_p), \ j' =
	(j'_1 , \ldots, j'_p)$ be $p$-tuples of positive integers from $\{1,2,\ldots, N\}$. 
Then,
	\[\int\limits_{\mathcal{U}(\bbC^N)} \! U_{i_1j_1}\cdots U_{i_pj_p} \overline{U_{i'_1j'_1}} \cdots \overline{U_{i'_pj'_p}} dU  = \sum_{\sigma, \tau \in S_p} \delta_{\sigma}(i,i')\delta_{\tau}(j,j')    \wg( \tau \sigma^{-1},N)\]
 where $\delta_{\sigma}(i,i') = \delta_{i_1i'_{\sigma(1)}} \delta_{i_2i'_{\sigma(2)}} \ldots \delta_{i_pi'_{\sigma(p)}}$ and $\delta$ is the standard Kronecker delta function.
 
	If $p \ne p'$, then 
	\[ \int\limits_{\mathcal{U}(\bbC^N)} \! U_{i_1j_1} \cdots U_{i_pj_p} \overline{U_{i'_1j'_1}} \cdots \overline{U_{i'_{p'}j'_{p'}}} dU  = 0.\]	
\end{fact}

The following result encloses all the information we need for our computations about the asymptotics of the $\wg$ function; see \cite{C03}  for a proof.

\begin{fact} [Asymptotics of Weingarten functions (Section 2.6.3, \cite{Gu13})]
For $\sigma \in S_{t}$, 
\begin{equation}\label{wg_asym}\wg (\sigma, N) = \clO \left(\dfrac{N^{\vert C(\sigma)\vert}}{N^{2t}}\right) \hspace{2cm} \text{ as } N \to \infty \;.
\end{equation}
\end{fact}

 \begin{fact}[Proposition~2.4, \cite{Gu13}]
For all $t \geq 1$, 
\begin{equation}\label{wg_sum} \sum_{\sigma \in S_{t}}\wg (\sigma, N) = \frac{1}{N(N+1) \cdots (N+t-1)} \;.
\end{equation} \end{fact}

Other than the sum of the Weingarten function, one more quantity that will be important for us is its $L_1$ norm. 
Here, we derive a useful expression for that.   
\begin{lemma}\label{wg_modsum}	
	For all $t \geq 1$, 	
\begin{equation} \sum_{\sigma \in S_{t}} \vert \wg (\sigma, N) \vert = \frac{1}{N(N-1) \cdots (N-(t-1))} \;.
\end{equation}	
\end{lemma}
\begin{proof}
	Let $\rho_{sign}$ denote the sign representation of the symmetric group. 
	Let $G$ denote the inverse of $\wg$ in $\mathbb{C}\left[S_t\right]$.
	It is well known that $G= \prod\limits_{k=1}^{t}(N+J_k)$ where $J_k$ is $k$-th \emph{Jucys-Murphy element}, defined as follows:	
\[ J_k= \left(1,2 \right) + \left(2,3 \right) + \cdots+ \left(k-1,k \right). \]		
Now $G= \prod\limits_{k=1}^{t}(N+J_k)$	gives
$$\rho_{sign}\left(G\right)= \rho_{sign}\left( \prod\limits_{k=1}^{t}(N+J_k)\right).$$
Inverting both sides, 
\begin{align*}
\rho_{sign}\left( \wg\right)&= \left(\rho_{sign}\left(\prod\limits_{k=1}^{t}(N+J_k)\right)\right)^{-1} \\
&= \left(\prod\limits_{k=1}^{t} \rho_{sign}(N+J_k)\right)^{-1}\\
&= \left( \prod\limits_{k=1}^{t}(N-(k-1))\right)^{-1}\\
&= \frac{1}{N (N-1) (N-2) (N-(t-1))} \,. & \qedhere
\end{align*}
	\end{proof}
We give some values of the Weingarten functions for the unitary group $\mathcal{U}(\bbC^N)$ taken from \cite{C03} upto third moments. 
\begin{align*}
\wg( [1], N) = \dfrac{1}{N} \;,  &  \hspace{1.6cm}
\wg( [1,1], N) = \dfrac{1}{N^2-1} \;,  \\
\wg( [2], N) = \dfrac{-1}{N(N^2-1)} \;, &  \hspace{1.6cm}
 \wg( [1,1,1], N) = \dfrac{N^2-2}{N(N^2-1)(N^2-4)} \;,  \\
 \wg( [2,1], N) = \dfrac{-1}{(N^2-1)(N^2-4)} \;, &   \hspace{1.6cm}
\wg( [3], N) = \dfrac{2}{N(N^2-1)(N^2-4)}\;. 
\end{align*}

\section{A Warm-up: Quantum tamper detection codes for classical messages}\label{cqtdc}
In this section, we consider quantum tamper detection codes for classical messages. We give a probabilistic proof that quantum tamper detection codes for classical messages exist.
\begin{theorem}\label{mainthm}
Let $\mathcal{U}_{\mathsf{Adv}}$ be an $\left(N,\alpha, \sqrt{\frac{\eps}{2K}}\right)$ family such that $ \left( \frac{1}{6} - \alpha\right) \log(N)  \geq \log(K) + \log\left( \frac{1}{\eps}\right)+ 2$.
Then there exists a $(K, N, \eps)$-tamper secure scheme for classical messages.
Furthermore, a uniformly random encoding and decoding strategy according to Haar measure \emph{(}see $(\enc, \dec_{\mathsf{Cl}})$ in Definition~\ref{def:random_haar_enc_dec}\emph{)} gives such a code with probability at least $1-  \mathcal{O} \left(\dfrac{K N}{ 2^{N^{\alpha}}}\right)$. 
\end{theorem}
\begin{proof}	
	We show that an encoding and decoding strategy, as given in Definition~\ref{def:random_haar_enc_dec}, gives a tamper detection code for the given set of parameters. 

	For a fixed unitary $U \in \mathcal{U}_{\mathsf{Adv}}$, let us define random variables $X_{js} = \vert \bra{\psi_j}U\ket{\psi_s}\vert^2$  {for $j,s \in \mathcal{M}$}.
	 Here the randomness is over the Haar measure in choosing $(\enc, \dec )$ strategy as an isometry $V$.
	 Let $X_s = \sum_{j \ne s}  X_{js}.$ 
	 The random variable $X_{js}$ denotes the probability that message $j$ was decoded given that message $s$ was encoded.
	 Similarly, $X_s$ denotes the probability that the procedure resulted in an incorrectly decoded message.  
	 Both $X_{js}$ and $X_s$ are non-negative random variables with values less than or equal to $ 1$.
	 
	Let $\clE$ be the event that  $(\enc, \dec)$  is not an $\eps$-secure tamper detection code against  $\mathcal{U}_{\mathsf{Adv}}$. Then, 	
	\begin{align*}
		\Pr\left(\clE   \right) & \leq \Pr\left(   \exists  U \in \mathcal{U}_{\mathsf{Adv}}, s \in \{ 0,1 \}^k  \quad \text{s.t.}  \quad X_s+X_{ss} \geq \eps \right)\\
		& \leq \sum_{ U \in \mathcal{U}_{\mathsf{Adv}}} \sum_{s \in \{ 0,1 \}^k} \Pr\left(  X_s +X_{ss}  \geq \eps  \right)  \\
		& \leq \sum_{ U \in \mathcal{U}_{\mathsf{Adv}}} \sum_{s \in \{ 0,1 \}^k} \sum_{j \in \{ 0,1 \}^k \ne s}  \Pr\left(  X_{js}  \geq \frac{\eps}{K}   \right) +\sum_{ U \in \mathcal{U}_{\mathsf{Adv}}} \sum_{s \in \{ 0,1 \}^k} \Pr\left(  X_{ss}  \geq \frac{\eps}{K}  \right)  \\
		& \leq  \vert \mathcal{U}_{\mathsf{Adv}} \vert K^2 \Pr\left( X_{js}  \geq \frac{\eps}{K}   \right) + \vert \mathcal{U}_{\mathsf{Adv}} \vert K \Pr\left(  X_{ss}  \geq \frac{\eps}{K}    \right) \\
		& =  \vert \mathcal{U}_{\mathsf{Adv}} \vert K^2 	\Pr\left(\clE_1   \right) + \vert \mathcal{U}_{\mathsf{Adv}} \vert K 	\Pr\left(\clE_2   \right). \\
	\end{align*}
	To bound  $\Pr\left(\clE_1   \right)=\Pr\left(  X_{js}  \geq \frac{\eps}{K}   \right) $ and  $\Pr\left(\clE_2   \right)=\Pr\left(  X_{ss}  \geq  \frac{\eps}{K}   \right)$ using a Chernoff-like argument, we need to calculate moments of random variable $X_{js}$ and $X_{ss}$.	
	Note that we could not directly use Chernoff bound to bound $\sum_j X_{js}$ as for different $j_1 \neq j_2$, the random variables $X_{j_{1}s}$ and $X_{j_{2}s}$ are not independent of each other.
	Naturally, the problem of calculating moments of random variable $X_{js}$ is closely related to Weingarten unitary calculus (see Section~\ref{wuc}) as our encoding strategy is Haar random.

Here we present first-order moments for variables $X_{js}$ and $X_{ss}$. 
Computation for higher moments is similar but slightly more involved and can be found in the Appendix~\ref{sec:app_classical}.

For readability we use $\phi = \sqrt{\frac{\eps}{2K}}$. 
Thus, $\mathcal{U}_{\mathsf{Adv}}$ is an $(N, \alpha, \phi)$ family.

\vspace{0.4cm}
 
	\textbf{First moment of random variable $X_{js}$ and $X_{ss}$:} 
	
	We begin with the first moment of $X_{js}$. 
	\begin{align*}
		X_{js} & = |\bra{\psi_j}U\ket{\psi_s}|^2\\
		& = \bra{\psi_j}U\ket{\psi_s}\bra{\psi_s}U^\dagger \ket{\psi_j} \\
		& = \bra{j}V^\dagger UV\ket{s}\bra{s}V^\dagger U^\dagger V\ket{j} \\
		& = \left(\sum_{l_1,k_1}U_{l_1k_1} V^*_{l_1j}V_{k_1s}\right) \left(\sum_{l_2 ,k_2}U^\dagger_{l_2k_2} V^*_{l_2s}V_{k_2j} \right) \\
		& = \sum_{l_1,k_1}\sum_{l_2,k_2}\left(U_{l_1k_1} U^\dagger_{l_2k_2} V_{k_1s}V_{k_2j} V^*_{l_1j}V^*_{l_2s}\right).
	\end{align*}
	
	\noindent\textbf{A.} When $j \ne s$,
	\begin{align*}
		\E[X_{js}]& = \E[|\bra{\psi_j}U\ket{\psi_s}|^2]\\
		& = \sum_{l_1,k_1}\sum_{l_2,k_2}\left(U_{l_1k_1} U^\dagger_{l_2k_2} \E\left[V_{k_1s}V_{k_2j} V^*_{l_1j}V^*_{l_2s}\right]\right) \\ 
		& = \sum_{l_1,k_1}\sum_{l_2,k_2}\left(U_{l_1k_1} U^\dagger_{l_2k_2}   \left( \sum_{\alpha, \beta \in S_2} \delta_{\alpha}(k_1k_2,l_1l_2)\delta_{\beta}(sj,js) \wg(\beta \alpha^{-1},N)\right) \right).
	\end{align*} 
	The final equality is due to Fact~\ref{wgt}.
	Note that when $j \ne s$ and $\beta= \id$, we get, $\delta_{\beta}(sj,js) =0$.
	Thus, the only terms that survive are those corresponding to $\beta= \left(1 \ 2 \right)$. 
	
	\begin{align*}
		\E[X_{js}]& = \sum_{l_1,k_1,l_2,k_2}U_{l_1k_1} U^\dagger_{l_2k_2}   \left( \delta(k_1k_2,l_1l_2) \wg((1\ 2) ((1)(2))^{-1},N) + \delta(k_1k_2,l_2l_1) \wg((1\ 2) (1\ 2)^{-1},N)  \right)\\
		& = \sum_{l_1=k_1}\sum_{l_2=k_2}U_{l_1k_1} U^\dagger_{l_2k_2}  \wg((1\ 2),N) +   \sum_{l_1=k_2}\sum_{l_2=k_1}U_{l_1k_1} U^\dagger_{l_2k_2} \wg((1)(2),N)  \\ 
		& = \Tr(U)  \Tr(U^\dagger)\cdot \wg((1\ 2),N) +   \Tr(UU^\dagger) \cdot \wg((1)(2),N)  \\ 
		& =\frac{- \Tr(U)  \Tr(U^\dagger) }{N(N^2-1)}+  N \cdot \frac{1}{N^2-1} \\ 
		& =\frac{ N^2- \Tr(U)  \Tr(U^\dagger) }{N(N^2-1)}\\ 
		& =\frac{ N^2- \vert\Tr(U)\vert^2  }{N(N^2-1)}\\ 
		& \leq \frac{ 2 }{N} \;.
	\end{align*}

 \noindent \textbf{B.} When $j = s$,
	\begin{align*}
		\E[X_{ss}]& = \E[|\bra{\psi_s}U\ket{\psi_s}|^2] \nonumber \\
		& = \sum_{l_1,k_1}\sum_{l_2,k_2}\left(U_{l_1k_1} U^\dagger_{l_2k_2} \E\left[V_{k_1s}V_{k_2s} V^*_{l_1s}V^*_{l_2s}\right]\right) \nonumber \\ 
		& = \sum_{l_1,k_1}\sum_{l_2,k_2}\left(U_{l_1k_1} U^\dagger_{l_2k_2}   \left( \sum_{\alpha, \beta \in S_2} \delta_{\alpha}(k_1k_2,l_1l_2)\delta_{\beta}(ss,ss) \wg(\beta \alpha^{-1},N)\right) \right) & (\mbox{from Fact~\ref{wgt}}) \\ 
		& =  \sum_{\alpha \in S_2} \left(\sum_{l_1,k_1,l_2,k_2}U_{l_1k_1} U^\dagger_{l_2k_2}    \delta_{\alpha}(k_1k_2,l_1l_2) \left(\sum_{\beta \in S_2}\wg(\beta \alpha^{-1},N)\right) \right) \nonumber \\
		& =  \sum_{\alpha \in S_2} \left(\sum_{l_1,k_1,l_2,k_2}U_{l_1k_1} U^\dagger_{l_2k_2}    \delta_{\alpha}(k_1k_2,l_1l_2) \left( \frac{1}{N(N+1)}\right) \right)& (\mbox{from eq.~\eqref{wg_sum}}) \\
		& = \frac{1}{N(N+1)}\left(\sum_{l_1=k_1}\sum_{l_2=k_2}U_{l_1k_1} U^\dagger_{l_2k_2}  +   \sum_{l_1=k_2}\sum_{l_2=k_1}U_{l_1k_1} U^\dagger_{l_2k_2}\right) \nonumber  \\ 
		& = \frac{\Tr(U)  \Tr(U^\dagger) +   \Tr(UU^\dagger) }{N(N+1)}  \\ 
			& =  \frac{ N+ \vert \Tr(U)\vert^2 }{N(N+1)} \nonumber \leq \phi^2 + \frac{1}{N} & (\mbox{since $\vert\Tr(U) \vert \leq \phi N$}).
	\end{align*}

Thus, we get the following bounds:
\begin{gather*}
     \E\left[X_{js}\right] \leq \frac{2}{N} \quad  and \quad 
    \E[X_{ss}] \leq \phi^2+ \frac{1}{N} \;.
\end{gather*}

Similarly, we get higher moment bounds (see Appendix~\ref{sec:app_classical});
\begin{gather}
    \E\left[X_{js}^{t}\right] \leq\mathcal{O}\left( \frac{t^4}{N}\right)
^t \quad and \quad \E[X_{ss}^{t}] \leq \mathcal{O}{\left( \left( \frac{t^2}{N}\right)^t + t \phi^{2t} \right)}. \label{claim:t-moments_classical}   \end{gather}

Now we proceed to bound the probability $ \Pr\left(  X_{js}  \geq  \frac{\eps }{K}  \right).$
	\begin{align}
		\Pr\left(  X_{js}  \geq  \frac{\eps }{K}   \right) & \leq \Pr\left(  e^{ \theta X_{js}}  \geq e^{ \frac{\theta \eps }{K}   }  \right) \nonumber  \\
		& \leq \frac{\E[  e^{ \theta X_{js}} ] }{e^{ \frac{\theta \eps }{K}  }  } \nonumber  \\
		& =  \frac{ 1 }{e^{ \frac{\theta \eps }{K}  } }  \sum_i  \frac{ \theta^i \E[X_{js}^i]}{i! } \nonumber  \\
		& = \frac{ 1 }{e^{ \frac{\theta \eps }{K}  } }  \left( \sum_{i=0}^{N^{1/5}}  \frac{\theta^i   \E( X_{js}^i)}{i! }+  \sum_{i \geq N^{1/5}+1} \frac{ \theta^i  \E(  X_{js}^i)}{i! }\right) \nonumber  \\
		& \leq \frac{ 1 }{e^{ \frac{\theta \eps }{K}    } } \left( \sum_{i=0}^{N^{1/5}}  \frac{ \theta^i}{i! } \clO \left(\frac{ i^4 }{N }\right)^i+  \sum_{i \geq N^{1/5}+1} \frac{\theta^i  \E(X_{js}^i)}{i! }\right)   &  (\mbox{from eq.~\eqref{claim:t-moments_classical}}) \nonumber \\
		& \leq \frac{ 1 }{e^{ \frac{\theta \eps }{K}  } } \left( \sum_{i=0}^{N^{1/5}}  \frac{ 1}{i! } \clO \left(\frac{ \theta}{N^{1/5} }\right)^i+  \sum_{i \geq N^{1/5}+1} \frac{\theta^i  \E(X_{js}^i)}{i! }\right) \nonumber  \\
		& \leq \frac{ 1 }{e^{ \frac{\theta \eps }{K} } } \left( \sum_{i=0}^{N^{1/5}}  \frac{ 1}{i! } \clO \left(\frac{ \theta}{N^{1/5} }\right)^i+  \sum_{i \geq N^{1/5}+1} \frac{\theta^i  \E(  X_{js}^{N^{1/5}})}{i! }\right) \nonumber  \\
		& \leq    \frac{ 1 }{e^{ \frac{\theta \eps }{K}  } }  \left( \clO \left(e^{ \frac{\theta}{N^{1/5}}} \right)+  \clO \left(\frac{e^\theta}{ N^{\frac{1}{5}N^{1/5}}  }\right)\right) \nonumber  \\
		& \leq    \frac{ 1 }{e^{ \frac{\theta \eps }{K} } } \clO \left( e^{ \frac{\theta}{N^{1/5}}} +  e^{\theta - \frac{1}{5}N^{1/5} \ln(N)  } \right) \nonumber  \\
		\nonumber &  \leq    \frac{ 1 }{e^{ \frac{ N^{1/6} \eps }{K} } } \clO \left( 1 +1 \right) & (\mbox{choosing $\theta = N^{\frac{1}{6}}$}) \\
 \label{jsboundeq10} & \leq \mathcal{O} \left( e^{ - \frac{N^{1/6} \eps }{K} }  \right).
	\end{align} 
	
	Similarly when $j=s$, we bound the probability $ \Pr\left(  X_{ss} \geq \frac{\eps}{K}  \right).$
	\begin{align}
		\Pr\left(  X_{ss}  \geq  \frac{\eps }{K}   \right) & \leq \Pr\left(  e^{ \theta X_{ss}}  \geq e^{ \frac{\theta \eps }{K}   }  \right) \nonumber  \\
		& \leq \frac{\E[  e^{ \theta X_{ss}} ] }{e^{ \frac{\theta \eps }{K}  }  } \nonumber  \\
		& =  \frac{ 1 }{e^{ \frac{\theta \eps }{K}  } }  \sum_i  \frac{ \theta^i \E[X_{ss}^i]}{i! } \nonumber  \\
		& = \frac{ 1 }{e^{ \frac{\theta \eps }{K}  } }  \left( \sum_{i=0}^{N^{1/5}}  \frac{\theta^i   \E( X_{ss}^i)}{i! }+  \sum_{i \geq N^{1/5}+1} \frac{ \theta^i  \E(  X_{ss}^i)}{i! }\right) \nonumber  \\
		& \leq \frac{ 1 }{e^{ \frac{\theta \eps }{K}    } }\  \clO\left( \sum_{i=0}^{N^{1/5}}  \frac{ \theta^i}{i! } \left( i\phi^{2i}  + \left( \frac{1}{N^{3/5}} \right)^i\right)+  \sum_{i \geq N^{1/5}+1} \frac{\theta^i  \E(X_{ss}^i)}{i! }\right) \hspace{1.236cm} \nonumber
 \\ \nonumber &\omit\hfill  (\mbox{from eq.~\eqref{claim:t-moments_classical}}) \nonumber \\
		& \leq \frac{ 1 }{e^{ \frac{\theta \eps }{K}  } }\  \clO \left( N\sum_{i=0}^{N^{1/5}}  \frac{ (\theta\phi^{2})^i}{i! } +  \sum_{i=0}^{N^{1/5}}    \frac{ 1}{i! }   \left( \frac{\theta}{N^{3/5}} \right)^i  +    \left(\left(  \frac{1}{N^{3/5}}   \right)^{N^{1/5}}  + N \phi^{N^{1/5}}   \right) e^{\theta }\right) \nonumber  \\
		& \leq   \frac{ 1 }{e^{ \frac{\theta \eps }{K}  } }\ \clO\left(   N e^{\theta\phi^{2}} + e^{ \frac{\theta}{N^{3/5}}} +  e^{\theta   - \frac{3}{5} N^{1/5} \ln(N)  } + e^{\theta +\ln(N)   -N^{1/5} \ln(1/\phi)  } \right) \nonumber \\
 & \leq  e^{ - \frac{ N^{1/6} \eps }{K}  } \ \clO\left(   N e^{N^{\frac{1}{6}}\phi^{2}} + 1 + 1 +1  \right) 
 \tab \tab  \left(\mbox{choosing $\theta = N^{\frac{1}{6}}$}\right) \nonumber \\ & \leq  e^{ - \frac{ N^{ \eps \frac{1}{6}} }{K}  } \ \clO\left(   N e^{ \frac{\eps N^{\frac{1}{6}}}{4K}} + 1 \right) \tab \tab \tab 
 \left( \mbox{\mbox{since $\phi = \sqrt{\frac{\eps}{4K}}$}} \right)
 \nonumber
  \\ 
& \leq  \mathcal{O}\left( N e^{ - \frac{3 \eps N^{\frac{1}{6}}}{4K}} \right)  
  \label{sstboundeq111} 
	\end{align} 

	\begin{flalign}	\vert\mathcal{U}_{\mathsf{Adv}} \vert  K^2\Pr\left(\clE_1   \right)  & = 	\vert\mathcal{U}_{\mathsf{Adv}} \vert K^2 \Pr\left(  X_{js}  \geq \frac{\eps}{K}   \right) \nonumber \\
		& \leq  \vert\mathcal{U}_{\mathsf{Adv}} \vert K^2\ \mathcal{O} \left( e^{ - \frac{N^{1/6} \eps }{K} }  \right) &(\mbox{from eq.~\eqref{jsboundeq10}}) \nonumber\\
		& \leq  2^{N^{\alpha}} K ^2   \mathcal{O} \left( e^{ - 4 N^{\alpha} } \right)   \nonumber \\
  & \leq   K ^2   \mathcal{O} \left( 2^{N^{\alpha}} 2^{ - 4 N^{\alpha} } \right)   \nonumber \\
				\label{e1eq4}& \leq \mathcal{O}\left(\dfrac{K^2}{ 2^{N^{\alpha}}}\right). 
	\end{flalign}
	The third inequality follows from the choice of our parameters; 
	\[ \left( \frac{1}{6} - \alpha\right) \log(N)  \geq \log(K) + \log\left( \frac{1}{\eps}\right)+ 2.\]

    Similarly, we have, 
	\begin{flalign}
		 \vert \mathcal{U}_{\mathsf{Adv}} \vert K \Pr\left(\clE_2   \right) & = \vert \mathcal{U}_{\mathsf{Adv}} \vert K \Pr\left(  X_{ss}  \geq \frac{\eps}{K}   \right) \nonumber \\
		\nonumber & \leq \vert \mathcal{U}_{\mathsf{Adv}}\vert  K   \mathcal{O}\left( N e^{ - \frac{3 \eps N^{\frac{1}{6}}}{4K}} \right)   & (\mbox{from eq.~\eqref{sstboundeq111}}) \\
		\nonumber & \leq 2^{N^{\alpha}} K  \mathcal{O} \left( N  e^{ - 3 N^{\alpha} } \right)   \\ & \leq  K  \mathcal{O} \left( N 2^{N^{\alpha}}  2^{ - 3 N^{\alpha} } \right)  \nonumber \\ 
		\label{e2eq4}& \leq   \mathcal{O} \left(\dfrac{K N}{ 2^{N^{\alpha}}}\right). 
	\end{flalign}
The third follows from our choice of parameters: $ \left( \frac{1}{6} - \alpha\right) \log(N)  \geq \log(K) + \log\left( \frac{1}{\eps}\right)+ 2$.
Thus, it follows from eq.~\eqref{e1eq4}~and~\eqref{e2eq4} that
\begin{align*}
	\Pr\left(\clE   \right) & \leq  \vert \mathcal{U}_{\mathsf{Adv}} \vert K^2 	\Pr\left(\clE_1   \right) + \vert \mathcal{U}_{\mathsf{Adv}} \vert K 	\Pr\left(\clE_2   \right) \leq  \mathcal{O} \left(\dfrac{K N}{ 2^{N^{\alpha}}}\right). \qedhere
\end{align*}
\end{proof}
\subsection{Relaxed tamper detection for classical messages}\label{Sec:Relaxed_TD}
We would like to point out that an interesting side result follows from our previous calculation. 
It follows that one can get a \emph{relaxed} version of tamper detection even if even when the family $\mathcal{U}_{\mathsf{Adv}}$ does not satisfy the \emph{far from identity} condition. 
Recall that, in the relaxed version, we aim to either output the original message or detect that it was tampered and output $\perp$.
In principle, the relaxed version allows us to revert back to the original message without detecting tampering. 
Such a ``reversion without detection" is inherent to the quantum setting due to the action of measurement operators.
For example, consider a message $m$ encoded as $\ket{\psi}$.
Suppose a unitary takes $\ket{\psi}$ to $\frac{1}{\sqrt{2}} \left( \ket{\psi} + \ket{\psi^\prime}\right)$ where $\ket{\psi^\prime}$ is orthogonal to the space of codewords.  The measurement of the decoder can result in $\ket{\psi^\prime}$ indicating that there was tampering.
If the measurement results in $\ket{\psi}$, we can not detect the tampering, but nonetheless, the decoder still outputs the correct message $\widehat{m}=m$. 
Thus, one gets a qualitatively similar version of tamper detection where the decoder either aborts or returns the correct plaintext.
\begin{theorem}\label{thm:relaxedclassical}
    Let $\mathcal{U}_{\mathsf{Adv}}$ be an $\left(N, \alpha, 1\right)$ family such that $ \left( \frac{1}{6} - \alpha\right) \log(N)  \geq \log(K) + \log\left( \frac{1}{\eps}\right)+ 2$. Then a uniform Haar random encoding-decoding strategy is $(K,N,\eps)$-relaxed tamper secure with probability at least $1-  \mathcal{O} \left(\dfrac{K^2}{2^{N^{\alpha}}}\right)$.
\end{theorem}
\begin{proof}For a fixed unitary $U$, recall that random variables were defined as follows: $X_{js} = |\bra{\psi_j}U\ket{\psi_s}|^2$ and $X_s = \sum\limits_{j \ne s}  X_{js}.$ Let $\clE$ be the event that  $(\enc, \dec)$  is not an $\eps$-secure relaxed tamper detection code against  $\mathcal{U}_{\mathsf{Adv}}$.
   \begin{align*}
			\Pr\left(\clE   \right) & \leq \Pr\left(   \exists  U \in \mathcal{U}_{\mathsf{Adv}}, s \in \{ 0,1 \}^k  \quad s.t.  \quad X_s  \geq \eps \right)\\
			& \leq \sum_{ U \in \mathcal{U}_{\mathsf{Adv}}} \sum_{s \in \{ 0,1 \}^k} \Pr\left(  X_s  \geq \eps  \right) \\
			& \leq \sum_{ U \in \mathcal{U}_{\mathsf{Adv}}} \sum_{s \in \{ 0,1 \}^k} \sum_{j \in \{ 0,1 \}^k \ne s}  \Pr\left(  X_{js}  \geq \frac{\eps}{K}   \right) \\
			& \leq  | \mathcal{U}_{\mathsf{Adv}} | K^2 \Pr\left(  X_{js}  \geq \frac{\eps}{K}   \right) \\
			& \leq \mathcal{O} \left(\dfrac{K^2}{2^{N^{\alpha}}}\right) &(\mbox{from eq.~\eqref{e1eq4}}). & \qedhere \\
		\end{align*}		
\end{proof}

\subsection*{From relaxed tamper detection to non-malleability}\label{Sec:nmcodes}
The relaxed form of tamper detection aims to either output the original message, or detect that it was tampered (indicated by the output $\perp$).
On the other hand, a non-malleable code insists that we either output the original message or an unrelated message, but with an additional requirement that the  probability (of a message being the same) depends only on the adversarial unitary $U$.
And hence, it is not a priori clear if relaxed tamper detection will immediately give non-malleable security.
In particular, the probability distribution may depend on $U$, as well as the original message $s$.
However, this potential dependency on $s$ can be removed by first analysing the distribution for an average $s$.
Then, a standard average-case to worst-case reduction shows that non-malleability can be achieved by incurring a nominal hit in the parameters.
This line of argument of first going to an average case setting to remove the dependency on $s$, followed by a reduction to worst case non-malleability is fairly common~(see for example, Section 3.3 in~\cite{BGJR23}). We include it below.

\begin{claim}\label{claimr:average_relaxed_TD}
Let $(\enc,\dec)$ be $\eps$-secure relaxed tamper detection scheme.
    Let $S$ be the uniform distribution on $\mathcal{M} = \lbrace 0,1 \rbrace^{k}$.
    Then, \[\dec \left( U \left( \enc(S)  \enc(S)^\dagger \right) U^\dagger  \right) \approx_{2 \eps} p_U S + (1-p_U)\perp,\]
    where $p_U= \frac{1}{2^k} \sum_s X_{ss}$.
\end{claim}
\begin{proof}
Note that, since $S$ is the uniform distribution,  each $s$ is sampled with probability $\frac{1}{2^k}$, and moreover, any particular $s$ gives back the same $s$ on decoding with probability $p_{\textsf{same}}(s)$, some different $s'$ with $p_{\textsf{diff}}(s)$ and $\perp$ with probability $p_{\perp}(s)$.
And hence, we can represent the relevant distribution as the following convex combination:
\[\dec \left( U \left( \enc(S)  \enc(S)^\dagger \right) U^\dagger  \right) = \frac{1}{2^k} \sum_{s} p_{\textsf{same}}(s) S + \frac{1}{2^k} \sum_{s} p_{\textsf{diff}}(s) S' + \frac{1}{2^k} \sum_{s} p_\perp(s) \perp.\]
Since $(\enc,\dec)$ is $\eps$-secure relaxed tamper detection code, $p_{\textsf{diff}}(s) \leq \eps$, for all $s$.

\noindent Thus, $\frac{1}{2^k} \sum_{s} p_{\textsf{diff}}(s) \leq \eps$.
The claim now follows by noting that $p_{\textsf{same}(s)} = X_{ss}$.
\end{proof}

\begin{theorem} \label{thm:non-malleability}
 Let $\mathcal{U}_{\mathsf{Adv}}$ be an $\left(N, \alpha, 1\right)$ family such that $ \left( \frac{1}{6} - \alpha\right) \log(N)  \geq 2\log(K) + \log\left( \frac{1}{\eps}\right)+ 2$. Then a uniform Haar random encoding-decoding strategy $(\enc,\dec)$ is a $2\eps$-secure non-malleable code (for classical messages against $\mathcal{U}_{\mathsf{Adv}}$) with probability at least $1-  \mathcal{O} \left(\dfrac{K^2}{2^{N^{\alpha}}}\right)$.
\end{theorem}
    \begin{proof}
Let $p_U = \frac{1}{2^k} \sum X_{ss}$ and $\eta = \perp$.
Set $\eps^\prime \leftarrow  \frac{\eps}{K}$\,.

\noindent Then, by choice of parameters, 
 $ \left( \frac{1}{6} - \alpha\right) \log(N)  \geq  \log(K) + \log\left( \frac{1}{\eps^\prime}\right)+ 2$.
 Hence, by Theorem~\ref{thm:relaxedclassical}, a Haar random encoding-decoding is $\eps^\prime$-secure relaxed tamper detection code with probability at $1- \mathcal{O}\left( \frac{K^2}{2^{N^{\alpha}}} \right)$.
Furthermore, by Claim~\ref{claimr:average_relaxed_TD}, 
\begin{equation}\label{eq:average_rtd_in_nonmalleability}
    \dec \left( U \left( \enc(S)  \enc(S)^\dagger \right) U^\dagger  \right) \approx_{2 \eps^\prime} p_U S + (1-p_U)  \perp \,.
\end{equation}
Now,
\begin{flalign*}
    & \Vert \dec \left( U \left( \enc(s)  \enc(s)^\dagger \right) U^\dagger  \right) - p_U s + (1-p_U) \perp\Vert_1 \\   \leq ~ & 2^k \cdot \Vert \dec \left( U \left( \enc(S)  \enc(S)^\dagger \right) U^\dagger  \right) - p_U S + (1-p_U) \perp\Vert_1 \\
     \leq ~ & 2^k \cdot 2 \eps^\prime & \hspace{1.236cm} (\mbox{from eq.~\eqref{eq:average_rtd_in_nonmalleability}})  \\
     \leq ~ &  2 \eps \,.  &  \qedhere 
\end{flalign*}
    \end{proof}

\section{Tamper Detection Codes for Quantum Messages\label{qtdc}}
		In this section, we consider quantum tamper detection codes for quantum messages.
  Again, we give a probabilistic proof that quantum tamper detection codes  exist for quantum messages.
		Our probabilistic methods are similar, but some subtle intricacies are involved for quantum messages due to superposition. 

\begin{theorem}\label{theorem:td_quantum_messages}
    Let $\mathcal{U}_{\mathsf{Adv}}$ be an $\left(N,\alpha, \sqrt{\frac{\eps}{2K}}\right)$ family such that $ \left( \frac{1}{6} - \alpha\right) \log(N)  \geq \log k + \log\left( \frac{1}{\eps}\right)+ 2$ and let $\delta = 2^{2 + \log K - \frac{N^\alpha}{K}}$.
Then a uniformly random Haar encoding and decoding strategy \emph{(}see $(\enc, \dec)$ in Definition~\ref{def:random_haar_enc_dec}\emph{)} is a $(K,N,\eps + \delta)$-tamper secure scheme with probability at least $1-  \mathcal{O} \left(\dfrac{K N}{ 2^{N^{\alpha}}}\right)$. 
\end{theorem}
\begin{proof}

Let 
$ \mathcal{M}=\{ \ket{\theta_1},  \ket{\theta_2} , \ldots ,  \ket{\theta_M}\}$ be a $\delta$-net of $\mathbb{S}^{K-1}$ from Definition~\ref{epsnets} such that $M \leq (\frac{4K}{\delta})^K$ and $\delta = 2^{2 + \log K - \frac{N^\alpha}{K}}$. Let $\vert \theta \rangle$ be an arbitrary quantum message from $\delta$-net.
We express $\theta$ in the computational basis with $a_i$ as coefficients; $\vert \theta \rangle= \sum_{m=1}^{K} a_m \vert m \rangle$.
Recall $\vert \psi_m \rangle = V \vert m \rangle$.

		For $m\in [K]$, let $X_m =\left( \bra{\psi_m}\; U\;  (\enc\ket{\theta})  (\enc\ket{\theta})^\dagger U^\dagger \ket{\psi_m}\right)$ and $X = \sum_{m}X_m$.
		
Note that for a fixed $U$,
\begin{equation}\label{X_m_value_expanded}
    X_{m} = \left(  \sum_{i=1}^K \sum_{j=1}^K a_ia^*_j  \bra{\psi_m} U  \ket{\psi_i} \bra{\psi_j}  U^\dagger \ket{\psi_m}\right).
\end{equation}
Recall that $\Pi$ is a projector on the space of codewords,  that is, $\Pi = \sum\limits_{i=1}^K \ketbra{\psi_i}{\psi_i}$.  

		\begin{flalign*}
			X &= \Tr \left({\Pi}\; U \enc(\ket{\theta}(\enc(\ket{\theta})^{\dagger} U^\dagger\right)
            \\ & =  \sum_{m=1}^K  X_m  \\
			& =  \sum_{m=1}^K  \left( \bra{\psi_m} U \left(  \sum_{i=1}^K \sum_{j=1}^K a_ia^*_j  \ket{\psi_i} \bra{\psi_j}  \right) U^\dagger \ket{\psi_m}\right) & (\mbox{from eq.~\eqref{X_m_value_expanded}})  \\
			& =  \sum_{m=1}^K  \left(  \sum_{i=1}^K \sum_{j=1}^K a_ia^*_j  \bra{\psi_m}\; U  \ket{\psi_i} \bra{\psi_j}  U^\dagger \ket{\psi_m}\right).
		\end{flalign*}		
	Let $\clE$ be the event that $(\enc, \dec)$ is not  $\eps$-secure against $\mathcal{U}_{\mathsf{Adv}}$.
	Again, for bounding the probability of $\mathcal{E}$, we need the higher moments of $X_m$, the calculation of which we defer to Appendix~\ref{sec:app_quantum}.	
	\begin{equation}\label{eq:t_moment_quantum}
	    \E[X^t_{m}] \leq \clO \left(    \left( \frac{  t^2}{N}\right)^{t}  + t \phi^{ 2t}  \right).	\end{equation}
	
	After this, an argument similar to the previous one (breaking sum into two parts; $t \leq N^{1/5}$ and $t> N^{1/5}$ followed by union bound over all messages and accounting for the size of $\vert \mathcal{U}_{\mathsf{Adv}}\vert$) directly can be applied. 
	For completeness, we provide it here.

	\begin{align*}
		\Pr\left(  X_m  \geq  \frac{\eps }{K}   \right) & \leq \Pr\left(  e^{ \theta X_m}  \geq e^{ \frac{\theta \eps }{K}   }  \right) \nonumber  \\
		& \leq \frac{\E[  e^{ \theta X_m} ] }{e^{ \frac{\theta \eps }{K}  }  } \nonumber  \\
		& =  \frac{ 1 }{e^{ \frac{\theta \eps }{K}  } }  \sum_i  \frac{ \theta^i \E[X_m^i]}{i! } \nonumber  \\
		& = \frac{ 1 }{e^{ \frac{\theta \eps }{K}  } }  \left( \sum_{i=0}^{N^{1/5}}  \frac{\theta^i   \E( X_m^i)}{i! }+  \sum_{i \geq N^{1/5}+1} \frac{ \theta^i  \E(  X_m^i)}{i! }\right) \nonumber  \\
		& \leq \frac{ 1 }{e^{ \frac{\theta \eps }{K}    } } \left( \sum_{i=0}^{N^{1/5}}  \frac{ \theta^i}{i! } \clO \left(    \left( \frac{  i^2}{N}\right)^{i}  + i \phi^{ 2i}  \right)
		+  \sum_{i \geq N^{1/5}+1} \frac{\theta^i  \E(X_m^i)}{i! }\right) \hspace{1.236cm}(\mbox{from eq.~\eqref{eq:t_moment_quantum}}) \nonumber \\
		& \leq \frac{ 1 }{e^{ \frac{\theta \eps }{K}    } }\  \clO\left( \sum_{i=0}^{N^{1/5}}  \frac{ \theta^i}{i! } \left( i\phi^{2i}  + \left( \frac{1}{N^{3/5}} \right)^i\right)+  \sum_{i \geq N^{1/5}+1} \frac{\theta^i  \E(X_{ss}^i)}{i! }\right)   \nonumber \\
		& \leq \frac{ 1 }{e^{ \frac{\theta \eps }{K}  } }\  \clO \left( N\sum_{i=0}^{N^{1/5}}  \frac{ (\theta\phi^{2})^i}{i! } +  \sum_{i=0}^{N^{1/5}}    \frac{ 1}{i! }   \left( \frac{\theta}{N^{3/5}} \right)^i  +    \left(\left(  \frac{1}{N^{3/5}}   \right)^{N^{1/5}}  + N \phi^{N^{1/5}}   \right) e^{\theta }\right) \nonumber  \\
		& \leq   \frac{ 1 }{e^{ \frac{\theta \eps }{K}  } }\ \clO\left(   N e^{\theta\phi^{2}} + e^{ \frac{\theta}{N^{3/5}}} +  e^{\theta   - \frac{3}{5} N^{1/5} \ln(N)  } + e^{\theta +\ln(N)   -N^{1/5} \ln(1/\phi)  } \right)
\\ & \leq  \mathcal{O}\left( N e^{ - \frac{3 \eps N^{\frac{1}{6}}}{4K}} \right).
	\end{align*} 
And finally, with the union bound,
		\begin{align*}
			\Pr\left(\clE   \right) & \leq \Pr\left(\exists  ~ U, \ket{\theta} \in \mathcal{U}_{\mathsf{Adv}} \times  \mathcal{M} \text{ such that } X \geq \eps \right)\\
			& \leq \sum_{ U \in \mathcal{U}_{\mathsf{Adv}}}\sum_{ \ket{\theta} \in \mathcal{M}} \Pr\left(  X \geq \eps  \right)  \\
			& \leq \sum_{ U \in \mathcal{U}_{\mathsf{Adv}}}\sum_{ \ket{\theta} \in \mathcal{M}}  \Pr\left(  \sum_{m=1}^K  X_m \geq \eps  \right)  \\
			& \leq \sum_{ U \in \mathcal{U}_{\mathsf{Adv}}} \sum_{ \ket{\theta} \in \mathcal{M}}  \sum_{m=1}^K \Pr\left(  X_m \geq  \frac{\eps}{K}  \right)  \\
			& \leq \vert \mathcal{U}_{\mathsf{Adv}} \vert \vert \mathcal{M} \vert  K \Pr\left(  X_m \geq  \frac{\eps}{K}  \right) \\
   & \leq 2^{N^{\alpha}}  2^{N^{\alpha}} K \mathcal{O}\left( N e^{ - \frac{3 \eps N^{\frac{1}{6}}}{4K}} \right) \\ 
			& \leq \mathcal{O} \left(\dfrac{K N}{ 2^{N^{\alpha}}}\right). 
   \qedhere
		\end{align*}	
	\end{proof}

\section{Conclusion and future work}\label{conclusion}	
	Our main result exhibits the existence of quantum tamper detection codes for large families of unitary operators of size upto $2^{2^{\alpha n}}$. Since the proof is probabilistic, one natural direction would be to give a constructive proof for quantum tamper detection codes.
 However, it should be noted that such efficient constructions are not known even against a classical adversary of such a large size.
 Typically, efficient constructions are known for families of size $2^{\poly(n)}$ in the CRS model.
 Hence, one has to first find out families  of relatively small size (and of some interest) against which tamper detection can be made efficient.
 We present one such example, the family of generalized Pauli operators.    
 There are other natural follow-up questions:
	\begin{itemize}
		\item 
  An arbitrary quantum adversary is capable of doing CPTP operations. Can we provide quantum tamper detection security for families of CPTP maps? 	
  As a first work in this line, we restrict ourselves to unitary tamperings. 
		\item Similar to the classical result of \cite{JW15}, can we obtain an efficient construction of tamper detection codes for an arbitrary family of unitary operators of size $2^{s(n)}$ where $s$ is an arbitrary polynomial in $n$?	
		\item 	Classically tamper detection codes exist for any $\alpha <1$. In the current work, we show the existence of unitary tamper detection codes for $\alpha< \frac{1}{6}$.
  Although we note that with careful optimization of parameters, the same analysis goes through for any $\alpha <\frac{1}{4}$, it will be interesting to see if we can get tamper detection codes for $\alpha \geq \frac{1}{4}$, possibly using some other techniques. 
		\item
  Classical tamper detection codes turned out to be an important component in the construction of classical non-malleable codes.
  Even in the case of unitary tamperings against classical messages, we show that tamper detection can lead to meaningful non-malleable guarantees.    
 It would be interesting to see if a similar approach can be taken for quantum messages as well.
	\end{itemize}



\section*{Acknowledgment}
{We thank Thiago Bergamaschi for the helpful discussions. 
The work of Naresh Goud Boddu was done while he was a PhD student at the Centre for Quantum Technologies (CQT), NUS, Singapore. 
This work is supported by the Prime Minister’s Office, Singapore and the Ministry of Education,
Singapore, under the Research Centres of Excellence program.}

\bibliographystyle{quantum}
\bibliography{References}

\newpage

\begin{appendix}

\section{Quantum AMD codes}\label{qamd}
	Let $\mathbb{F}_q$ be the field of size $q$ with characteristic $p$.
	Let $d$ be an integer such that $p$ does not divide $d + 2$. Consider the following function $f : \mathbb{F}_q^d \times \mathbb{F}_q \to \mathbb{F}_q$ defined by \[f(s_1,s_2,\ldots, s_d,r) = \sum_{i=1}^{d}s_i r^i + r^{d+2}.\]	
	We consider an encoding and decoding strategy analogous to classical encoding~\cite{CDFPW08}. 
	The analysis and proof also follow similar lines and are fairly straightforward.
 Here we present the same for the sake of completeness.
	For compactness, we will use $s$ to denote $(s_1,s_2,\ldots,s_d)\in \mathbb{F}_q^d$.
	We will also use $v_{i:j}$ to denote the restriction of the vector $v$ to coordinates from $i$  through $j$.
That is, for a vector $v= \left( v_1,v_2, \ldots, v_n  \right)$, the restriction $v_{i:j} = (v_i, v_{i+1}, \ldots, v_{j})$. 	
	\begin{itemize}
		\item Let $\enc$ be a quantum encoding defined as below:
		\[\enc :V \ket{(s_1, s_2,\ldots, s_d)} \to \ket{\psi_{s}} = \frac{1}{\sqrt{q}} \sum_{r \in [q]}\ket{s, r, f(s, r)}.\]
				\item Let $\dec$ be the POVM $\lbrace \Pi_\perp, \Pi_{s \in \mathbb{F}_q^d} \rbrace$ such that
		\[{\Pi}_{s} = \ket{\psi_{s}} \bra{\psi_{s}} \text{ and } {\Pi}_\perp = \Id - \sum_{s \in \mathbb{F}_q^d} {\Pi}_{s} \,.\] 
		\end{itemize} 
			\begin{claim} \label{claim:AMD_inequality}
		$\left\vert \sum\limits_{r \in \mathbb{F}_q}  \bra{f((s + x_{1:d}),r+x_{d+1})}\ket{f(s, r)+x_{d+2}}\right\vert^2 \leq (d+1)^2$\,.
		\end{claim}
		\begin{proof}
		Note that the following equation
		\[ \sum_{i=1}^{d}(s_i+x_i)(r+x_{d+1})^i + (r+x_{d+1})^{d+2}= \sum_{i=1}^{d}s_i r^i + r^{d+2}+x_{d+2} \]	
		gives a $d+1$ degree polynomial in $r$.
		Hence, for at most $d+1$ values of $r$, we can get $f((s + x_{1:d}),r+x_{d+1}) = f(s, r)+x_{[d+2]}$.
		The desired inequality now follows.
	\end{proof}
	
	\begin{theorem}
		The above  $(\enc,\dec)$ construction is quantum tamper secure (in the relaxed form) against generalized Pauli matrices with parameters  $\left(d\log q , (d+2)\log q, \left( \frac{d+1}{q} \right)^2 \right)$. 
	\end{theorem}

	\begin{proof}
		Let the error term due to generalized Pauli unitary $X$ be $x = (x_1, x_2, \ldots, x_{d+2})$ to indicate the tampering by 
		
		\[ X^x= X^{x_1} \otimes X^{x_2} \otimes \cdots \otimes X^{x_{d+2}} .\] Similarly let the error term due to generalized Pauli unitary $Z$ be $z = (z_1, z_2, \cdots, z_{d+2})$ to indicate the tampering by 
		
		\[ Z^z= Z^{z_1} \otimes Z^{z_2} \otimes \cdots \otimes Z^{z_{d+2}} .\]
		
		For any message $s = (s_1, s_2, \ldots, s_d)$, the state of the message after encoding and the tampering operation is 
		
		\[ X^xZ^z \ket{\psi_{s}} = \frac{1}{\sqrt{q}} \sum_{r \in [q]}  \omega^{\langle z_{1:d}, s \rangle + z_{d+1}r + z_{d+2}f(r,s)}\ket{(s_1+x_1, \ldots, s_d+x_d), r+x_{d+1}, f(s_1, \ldots,s_d, r)+x_{d+2}}. \] 
		
		For any other message $s' = (s'_1, s'_2, \ldots, s'_d) \ne s$, the probability of outputting $s'$ when the encoded message  $s$ is tampered by $X^xZ^z$  is given by the probability $|\bra{\psi_{s'}}X^xZ^z \ket{\psi_{s}}|^2.$ 
		Thus, the probability of outputting a different message can be bounded as follows: 
		
		\begin{flalign*}
			&\ \sum_{s' \ne s} \vert \bra{\psi_{s'}}X^xZ^z \ket{\psi_{s}}|^2 \\
			 & =  \sum_{s' \neq s}  \left\vert \frac{1}{q} \sum_{r,r' \in [q]}  \omega^{<z_{1:d}, s> + z_{d+1}r + z_{d+2}f(r,s)}  \left\langle{s',r',f(s',r')}\vert{s + x_{1:d}, r+x_{d+1}, f(s, r)+x_{d+2}}\right\rangle\right\vert^2  \\
			& = \left\vert \frac{1}{q} \sum_{r \in [q]}  \omega^{<z_{1:d}, s> + z_{d+1}r + z_{d+2}f(r,s)} \bra{f((s + x_{1:d}),r+x_{d+1})}\ket{f(s, r)+x_{d+2}}\right\vert^2    \\
			& \leq   \left\vert  \frac{1}{q} \sum_{r \in [q]}  \bra{f((s + x_{1:d}),r+x_{d+1})}\ket{f(s, r)+x_{d+2}}\right\vert^2  \\
			& \leq \left(\frac{d+1}{q} \right)^2 \hspace{10cm}  (\mbox{from Claim~\ref{claim:AMD_inequality}}). \nonumber \qedhere
		\end{flalign*}
  \end{proof}		

\newpage 

\section{Higher moments for classical messages \label{sec:app_classical}}
	\begin{changemargin}{-1.123cm}{-1.123cm}
	Similar to the case of the first-order moments, we start expressing $X_{js}$ as a sum of products. 
We then deal with both the cases $j=s$ and $j \neq s$ individually. 
	    \vspace{0.4cm}	
	
\noindent		\textbf{Higher moments of random variable $X_{js}$ and $X_{ss}$:}\label{highmoments} 
	\begin{flalign*}
		& X^t_{js}  = \vert\bra{\psi_j}U\ket{\psi_s}\vert^{2t}\\
		& = \left(\bra{\psi_j}U\ket{\psi_s}\bra{\psi_s}U^\dagger \ket{\psi_j}\right)^t \\
		& = \left(\sum_{l_1,k_1}U_{l_1k_1} V^*_{l_1j}V_{k_1s}\right) \left(\sum_{l_2 ,k_2}U^\dagger_{l_2k_2} V^*_{l_2s}V_{k_2j} \right) \cdots \left(\sum_{l_{2t-1},k_{2t-1}}U_{l_{2t-1}k_{2t-1}} V^*_{l_{2t-1}j}V_{k_{2t-1}s}\right) \left(\sum_{l_{2t} ,k_{2t}}U^\dagger_{l_{2t}k_{2t}} V^*_{l_{2t}s}V_{k_{2t}j} \right)  \\
		& = \sum_{l_1,k_1}\sum_{l_2,k_2}\cdots\sum_{l_{2t-1},k_{2t-1}}\sum_{l_{2t},k_{2t}}\left(U_{l_1k_1} U^\dagger_{l_2k_2}\ldots U_{l_{2t-1}k_{2t-1}} U^\dagger_{l_{2t}k_{2t}}  V_{k_1s}V_{k_2j} \ldots V_{k_{2t-1}s}V_{k_{2t}j}V^*_{l_1j}V^*_{l_2s}\ldots V^*_{l_{2t-1}j}V^*_{l_{2t}s}\right). 
	\end{flalign*}
	Before going ahead, we would like to introduce some shorthand and notation, given the number of terms involved in expressions to come.
	\begin{definition} \label{def:Upower_Ci}
	For a unitary operator, let $U^{c_i}$ be defined as follows:
	
		 \hspace{2cm} $U^{c_i}=U$ if $c_i$ is odd and
		 
		 \hspace{2cm} $U^{c_i}=U^\dagger$ if $c_i$ is even.
	\end{definition}	 
		 For definitions of $C(\alpha), C_1(\alpha), \Sigma_i$ and $\sfV(\alpha)$ see Section~{\ref{permu_group}}. 	
   See Section~\ref{wuc} for the definition of $\delta$ as well as other notations regarding Weingarten functions.	

   \vspace{1cm}
   
\noindent	\textbf{A.} When $j \ne s$,
	\begin{flalign}
		&\E[X^t_{js}] = \E[\vert\bra{\psi_j}U\ket{\psi_s}\vert^{2t}] \nonumber \\  & =  \sum_{l_1,k_1}\cdots \sum_{l_{2t},k_{2t}}\left(U_{l_1k_1} U^\dagger_{l_2k_2}\ldots U_{l_{2t-1}k_{2t-1}} U^\dagger_{l_{2t}k_{2t}} \E \left[ V_{k_1s}V_{k_2j} \ldots V_{k_{2t-1}s}V_{k_{2t}j}V^*_{l_1j}V^*_{l_2s} \ldots V^*_{l_{2t-1}j}V^*_{l_{2t}s}\right]\right)  \nonumber \\ 
		& =  \sum_{l_1,k_1} \cdots \sum_{l_{2t},k_{2t}}U_{l_1k_1} U^\dagger_{l_2k_2} \ldots U_{l_{2t-1}k_{2t-1}} U^\dagger_{l_{2t}k_{2t}}  \left(\sum_{\alpha, \beta \in S_{2t}} \delta_{\alpha}(k_1 \ldots k_{2t},l_1 \ldots l_{2t})\delta_{\beta}(sj \ldots sj,js \ldots js) \wg(\beta \alpha^{-1},N)\right) \nonumber \\ 
		& =   \sum_{\alpha \in S_{2t}} \left[\left(\sum_{k_1=l_{\alpha(1)}}\cdots \sum_{k_{2t}=l_{\alpha(2t)}}U_{l_1k_1} U^\dagger_{l_2k_2}\ldots U_{l_{2t-1}k_{2t-1}} U^\dagger_{l_{2t}k_{2t}} \right) \left(\sum_{ \beta \in S_{2t}}\delta_{\beta}(sj\ldots sj,js\ldots js) \wg(\beta \alpha^{-1},N)\right) \right] \nonumber \\
	\nonumber & =   \sum_{\alpha \in S_{2t}} \left[\left( \prod_{(c_1\ c_2\ \ldots\  c_e)\ \in\ C(\alpha)} \Tr(U^{c_1}U^{c_2}\cdots U^{c_e})\right) \left(\sum_{ \beta \in S_{2t}}\delta_{\beta}(sj\ldots sj,js \ldots js)\wg(\beta \alpha^{-1}, N)\right) \right]   \\ & \omit\hfill (\mbox{from Definition~\ref{def:Upower_Ci}}) \nonumber \\ 
		& =   \sum_{\alpha \in S_{2t}} \sum_{ \beta \in S_{2t}} \left[\left( \prod_{(c_1\ c_2\ \ldots\  c_e)\ \in\ C(\alpha)} \Tr(U^{c_1}U^{c_2}\ldots U^{c_e})\right) \left(\delta_{\beta}(sj\ldots sj,js \ldots js) \wg(\beta \alpha^{-1},N)\right) \right] \nonumber \\ 
		& \leq   \sum_{\alpha \in S_{2t}} \sum_{ \beta \in S_{2t}} \left[\left( \prod_{(c_1\ c_2\ \ldots\  c_e)\ \in\ C(\alpha)}  \vert \Tr(U^{\sfV(c_1 \ c_2 \ldots c_e)})   \vert \right) \left(\delta_{\beta}(sj\ldots sj,js \ldots js) \vert \wg(\beta \alpha^{-1},N)\vert\right) \right] \nonumber \\ 
		\nonumber & \leq   \sum_{\alpha \in S_{2t}} \sum_{ \beta \in S_{2t}} \left[\left( \prod_{(c_1\ c_2\ \ldots\  c_e)\ \in\ C(\alpha)}  N \right)  \left(\delta_{\beta}(sj\ldots sj,js \ldots js) \vert \wg(\beta \alpha^{-1},N)\vert\right) \right] \hspace{1.6cm}   (\mbox{from Fact~\ref{fact:unitary_trace_leq_n}}) \\
		\nonumber& \leq  \sum_{\alpha \in S_{2t}} \sum_{ \beta \in S_{2t}} \left[ N^{\vert C(\alpha) \vert} \clO\left(\delta_{\beta}(sj\ldots sj,js \ldots js)  \frac{N^{\vert C(\beta \alpha^{-1})\vert} }{N^{4t}}\right) \right] \hspace{1.6cm} \hspace{2.8cm}   (\mbox{from eq.~\eqref{wg_asym}})\\  
		& =  \sum_{\alpha \in S_{2t}} \sum_{ \beta \in S_{2t}} \left[ \clO \left(\delta_{\beta}(sj\ldots sj,js \ldots js)  \frac{N^{ \vert  C(\alpha) \vert+\vert C(\beta \alpha^{-1})\vert} }{N^{4t}}\right) \right] \nonumber \\ 
		& =  \sum_{\alpha \in S_{2t}} \sum_{ \beta \in \mathcal{B}_{2t}} \clO \left[ \frac{N^{ \vert C(\alpha) \vert+\vert C(\beta \alpha^{-1})\vert} }{N^{4t}} \right] \nonumber \\ 
		\nonumber & \leq  \sum_{\alpha \in S_{2t}} \sum_{ \beta \in \mathcal{B}_{2t}} \left[ \clO \left(\frac{1}{N^t} \right)  \right] \hspace{7.6cm} \tab (\mbox{from Corollary~\ref{upper_bound_on_C}}) \\ 
		& \leq  (2t)!  (t)! \left[ \clO \left(\frac{1}{N^t} \right)  \right] \nonumber \\ 
		& \leq   \clO \left(\frac{t^4}{N}\right)^t \hspace{10.6cm} \tab (\mbox{from Fact~\ref{ub_on_factorial}})\,. \nonumber
	\end{flalign} 

 \vspace{1cm}
 
\noindent		\textbf{B.} When $j = s$,
	\begin{flalign}
		&\E\left[X^t_{ss}\right]= \E\left[ \left\vert\bra{\psi_s}U\ket{\psi_s}\right\vert^{2t}\right] \nonumber\\ 
		& =  \sum_{l_1,k_1}\cdots \sum_{l_{2t},k_{2t}}\left(U_{l_1k_1} U^\dagger_{l_2k_2}\ldots U_{l_{2t-1}k_{2t-1}} U^\dagger_{l_{2t}k_{2t}} \E \left[ V_{k_1s}V_{k_2s} \ldots V_{k_{2t-1}s}V_{k_{2t}s}V^*_{l_1s}V^*_{l_2s} \ldots V^*_{l_{2t-1}s}V^*_{l_{2t}s}\right]\right)  \nonumber \\ 
		& =  \sum_{l_1,k_1} \cdots \sum_{l_{2t},k_{2t}}U_{l_1k_1} U^\dagger_{l_2k_2} \ldots U_{l_{2t-1}k_{2t-1}} U^\dagger_{l_{2t}k_{2t}}  \left(\sum_{\alpha, \beta \in S_{2t}} \delta_{\alpha}(k_1 \ldots k_{2t},l_1 \ldots l_{2t})\delta_{\beta}(ss \ldots ss,ss \ldots ss) \wg(\beta \alpha^{-1},N)\right) \nonumber \\ 
		& =   \sum_{\alpha \in S_{2t}} \left[\left(\sum_{k_1=l_{\alpha(1)}}\cdots \sum_{k_{2t}=l_{\alpha(2t)}}U_{l_1k_1} U^\dagger_{l_2k_2}\ldots U_{l_{2t-1}k_{2t-1}} U^\dagger_{l_{2t}k_{2t}} \right) \left(\sum_{ \beta \in S_{2t}}\delta_{\beta}(ss\ldots ss,ss\ldots ss) \wg(\beta \alpha^{-1},N)\right) \right] \nonumber \\
		\nonumber & =   \sum_{\alpha \in S_{2t}} \left[\left( \prod_{(c_1\ c_2 \ldots  c_e) \in C(\alpha)} \Tr(U^{c_1}U^{c_2}\cdots U^{c_e})\right) \left( \frac{1}{N(N+1) \cdots (N+2t-1)}\right) \right] \hspace{1.23cm}(\mbox{from Definition~\ref{def:Upower_Ci}}) \\ 
		& \leq   \sum_{\alpha \in S_{2t}} \left[\left( \prod_{c \in C(\alpha)} \vert \Tr(U^{\sfV(c)}) \vert \right) \left( \frac{1}{N(N+1) \cdots (N+2t-1)}\right) \right] \nonumber \\ 
		&=   \sum_{\alpha \in S_{2t}} \left[\left( \prod_{c \in C_1(\alpha)} \vert \Tr(U) \vert  \prod_{c \in C(\alpha) \setminus C_1(\alpha)} \vert \Tr(U^{\sfV(c)}) \vert \right) \left( \frac{1}{N(N+1) \cdots (N+2t-1)}\right) \right]  \nonumber \\ 
	& \leq   \sum_{\alpha \in S_{2t}} \left[\left( (\phi N)^{\vert C_1(\alpha) \vert}  N^{\vert C(\alpha) \vert-\vert C_1(\alpha) \vert} \right) \left( \frac{1}{N(N+1) \cdots (N+2t-1)}\right) \right] \hspace{3.8cm}(\mbox{from Fact~\ref{fact:unitary_trace_leq_n}})  \nonumber  \\ 
		\nonumber& \leq   \sum_{\alpha \in S_{2t}} \left[\left( \phi^{\vert \fix(\alpha) \vert}  N^{\vert C(\alpha) } \right) \left( \frac{1}{(2t)! {\binom{N+2t-1}{2t}}}\right) \right] \hspace{6cm} ~~ (\mbox{since $\mathsf{Fix}(\alpha) \subseteq C_1(\alpha)$})   \\ 	
		& =   \left( \frac{1}{(2t)! {\binom{N+2t-1}{2t}}}\right) \sum_{i=1}^{2t} \left[  \sum_{ \alpha \in S_{2t} : C(\alpha)=i} \left( \phi^{\vert \fix(\alpha) \vert}  N^{\vert C(\alpha) \vert} \right)  \right]  \nonumber \\ 	
		\nonumber& \leq   \left( \frac{1}{(2t)! {\binom{N+2t-1}{2t}}}\right) \left( \sum_{i=1}^{t-1} \left[  \sum_{\alpha \in S_{2t} : C(\alpha)=i}  N^{i}  \right] +\sum_{i=t}^{2t} \left[  \sum_{\alpha \in S_{2t} : C(\alpha)=i} \left( \phi^{ 2i-2t}  N^{ i } \right)  \right]\right)  ~ \hspace{1.236cm} (\mbox{from Lemma~\ref{lemma_fix}})  \\ 	
		& =   \left( \frac{1}{(2t)! {\binom{N+2t-1}{2t}}}\right) \left( \sum_{i=1}^{t-1} \left[  \vert \Sigma_{2t-i} \vert  N^{i}  \right] +\sum_{i=t}^{2t} \left[   \vert \Sigma_{2t-i} \vert \left( \phi^{ 2i-2t}  N^{ i } \right)  \right]\right)  \nonumber \\ 	
		\nonumber& \leq   \left( \frac{1}{(2t)! {\binom{N+2t-1}{2t}}}\right) \left( \sum_{i=1}^{t-1} \left[ {\binom{2t}{2}}^{2t-i}  N^{i}  \right] +\sum_{i=t}^{2t} \left[  {\binom{2t}{2}}^{2t-i}\left( \phi^{ 2i-2t}  N^{ i } \right)  \right]\right)  ~ ~ \hspace{1.23cm} (\mbox{from Observation~\ref{upper_bound_on_Sigma}})\\ 	
		& =   \left( \frac{{\binom{2t}{2}}^{2t}}{(2t)! {\binom{N+2t-1}{2t}}}\right) \left( \sum_{i=1}^{t-1} \left[ \left(\frac{N}{{\binom{2t} {2}}}\right)^i \right] + \frac{1}{\phi^{2t}} \sum_{i=t}^{2t} \left[ \left( \frac{ \phi^{ 2}N}{ {\binom{2t} {2}}} \right)^i  \right]\right) \nonumber	
	\\ 
		\nonumber& \leq   \left( \frac{ e^{2t-1}(e^2t^2)^{2t}(2t)^{2t}}{(2t)^{2t} (N+2t-1)^{2t} }\right) \left( \sum_{i=1}^{t-1} \left[ \left(\frac{N}{t^2}\right)^i \right] + \frac{1}{\phi^{2t}} \sum_{i=t}^{2t} \left[ \left( \frac{ \phi^{ 2}N}{ t^2} \right)^i  \right]\right) ~ \hspace{3.8cm}(\mbox{from Fact~\ref{ub_on_factorial}})\\
		& =   \frac{1}{e}\left( \frac{ e^3t^2}{N+2t-1}\right)^{2t} \left( \sum_{i=1}^{t-1} \left[ \left(\frac{N}{t^2}\right)^i \right] + \frac{1}{\phi^{2t}} \sum_{i=t}^{2t} \left[ \left( \frac{ \phi^{ 2}N}{ t^2} \right)^i  \right]\right) \nonumber \\ 	
		\nonumber& \leq   \frac{1}{e}\left( \frac{ e^3t^2}{N+2t-1}\right)^{2t} \left( 2 \left[ \left(\frac{\sqrt{N}}{t}\right)^{2t} \right] + \frac{t}{\phi^{2t}}  \left[ \left( \frac{ \phi^{ 2}N}{ t^2} \right)^{2t} \right]\right) ~ \hspace{4.6cm}(\mbox{since $4t^2 \leq N$})\\ 	
		& \leq \clO \left(    \left( \frac{  t \sqrt{N}}{N+2t-1}\right)^{2t}  + \frac{t}{\phi^{2t}}  \left( \frac{ \phi^2 N}{N+2t-1}\right)^{2t}  \right)  \nonumber \\ 	
		\nonumber& \leq \clO \left(    \left( \frac{  t^2}{N}\right)^{t}  + t \phi^{ 2t}  \right).
	\end{flalign}

	\end{changemargin} 
	
	\newpage
 
    \section{Higher moments for quantum messages \label{sec:app_quantum}}
    
	\begin{changemargin}{-1.123cm}{-1.123cm}
	We start by representing $X_m^t$ as a sum of products and then move on to calculating higher moments.
	
	    \vspace{0.4cm}

\noindent	\textbf{Higher moments of random variable $X_{m}$:}\label{highmomentsofxm} 
			\begin{flalign*}
			&X^t_{m} = \left( \sum_{i=1}^K \sum_{j=1}^K a_ia^*_j  \bra{\psi_m} U  \ket{\psi_i} \bra{\psi_j}  U^\dagger \ket{\psi_m} \right)^{t}\\
			 &= \left( \sum_{i=1}^K \sum_{j=1}^K a_ia^*_j  \left(\sum_{l_1,k_1}U_{l_1k_1} V^*_{l_1m}V_{k_1i}\right)\left(\sum_{l_1,k_1}U^\dagger_{l_2k_2} V^*_{l_2j}V_{k_2m}\right) \right)^{t}\\
			 &= \left( \sum_{i_1,j_1, \ldots i_t, j_t=1}^K  a_{i_1} \ldots a_{i_t}a^*_{j_1} \ldots a^*_{j_t} \right. \\
			  & \quad \quad \left. \left( \sum_{l_1,k_1, \ldots l_{2t},k_{2t}} \left(U_{l_1k_1} U^\dagger_{l_2k_2}\ldots U_{l_{2t-1}k_{2t-1}} U^\dagger_{l_{2t}k_{2t}}  V_{k_1i_1}V_{k_2m} \ldots V_{k_{2t-1}i_t}V_{k_{2t}m}V^*_{l_1m}V^*_{l_2j_1}\ldots V^*_{l_{2t-1}m}V^*_{l_{2t}j_t}\right)  \right)\right).
		\end{flalign*}
		Thus, 
		\begin{flalign*} 
		&\E[X^t_{m}] =  \sum_{i_1,j_1=1}^K \ldots \sum_{i_t,j_t=1}^K a_{i_1} \ldots a_{i_t}a^*_{j_1} \ldots a^*_{j_t} 	\left( \sum_{l_1,k_1}\sum_{l_2,k_2}\cdots\sum_{l_{2t-1},k_{2t-1}}\sum_{l_{2t},k_{2t}}  \right.  \\
			&  \left. \quad \quad \quad \left(U_{l_1k_1} U^\dagger_{l_2k_2}\ldots U_{l_{2t-1}k_{2t-1}} U^\dagger_{l_{2t}k_{2t}}  \E \left[ V_{k_1i_1}V_{k_2m} \ldots V_{k_{2t-1}i_t}V_{k_{2t}m}V^*_{l_1m}V^*_{l_2j_1}\ldots V^*_{l_{2t-1}m}V^*_{l_{2t}j_t} \right] \right)  \right) \nonumber \\
			&=  \sum_{i_1,j_1=1}^K \ldots \sum_{i_t,j_t=1}^K a_{i_1} \ldots a_{i_t}a^*_{j_1} \ldots a^*_{j_t}
			\left( \sum_{l_1,k_1}\sum_{l_2,k_2}\cdots\sum_{l_{2t-1},k_{2t-1}}\sum_{l_{2t},k_{2t}}   \right.  \\
			&  \left. \quad \left(U_{l_1k_1} U^\dagger_{l_2k_2}\ldots U_{l_{2t-1}k_{2t-1}} U^\dagger_{l_{2t}k_{2t}}  \left(\sum_{\alpha, \beta \in S_{2t}} \delta_{\alpha}(k_1 \ldots k_{2t},l_1 \ldots l_{2t})\delta_{\beta}(i_1m \ldots i_tm,mj_1 \ldots mj_t) \wg(\beta \alpha^{-1},N)\right) \right)  \right) \nonumber \\
			& =   \sum_{\alpha \in S_{2t}} \left[\left(\sum_{k_1=l_{\alpha(1)}}\cdots \sum_{k_{2t}=l_{\alpha(2t)}}U_{l_1k_1} U^\dagger_{l_2k_2}\ldots U_{l_{2t-1}k_{2t-1}} U^\dagger_{l_{2t}k_{2t}} \right)  \right.  \\
			& \left. \quad \quad \quad   \left(\sum_{ \beta \in S_{2t}} \wg(\beta \alpha^{-1},N) \left[       \sum_{i_1, \cdots i_t,j_1, \cdots j_t=1}^K a_{i_1} \ldots a_{i_t}a^*_{j_1} \ldots a^*_{j_t}    \delta_{\beta}(i_1m\ldots i_tm,mj_1\ldots mj_t)        \right]\right) \right] \nonumber \\ 
			& =   \sum_{\alpha \in S_{2t}} \left[\left(\sum_{k_1=l_{\alpha(1)}}\cdots \sum_{k_{2t}=l_{\alpha(2t)}}U_{l_1k_1} U^\dagger_{l_2k_2}\ldots U_{l_{2t-1}k_{2t-1}} U^\dagger_{l_{2t}k_{2t}} \right) \left(\sum_{ \beta \in S_{2t}} \wg(\beta \alpha^{-1},N)  \vert a_m \vert^{2l(\beta)}  \right) \right] \nonumber \\
		    	\nonumber & =   \sum_{\alpha \in S_{2t}} \left[\left( \prod_{(c_1\ c_2 \ldots  c_e) \in C(\alpha)} \Tr(U^{c_1}U^{c_2}\cdots U^{c_e})\right)  \left(\sum_{ \beta \in S_{2t}} \wg(\beta \alpha^{-1},N)  \vert a_m \vert^{2l(\beta)}  \right) \right]   \\ 
			& \leq   \sum_{\alpha \in S_{2t}} \left[\left( \prod_{c \in C(\alpha)} \vert \Tr(U^{\sfV(c)}) \vert \right) \left(\sum_{ \beta \in S_{2t}} \vert \wg(\beta \alpha^{-1},N) \vert \right) \right] \nonumber \\ 
			&=   \sum_{\alpha \in S_{2t}} \left[\left( \prod_{c \in C_1(\alpha)} \vert \Tr(U) \vert  \prod_{c \in C(\alpha) \setminus C_1(\alpha)} \vert \Tr(U^{\sfV(c)}) \vert \right) \left( \frac{1}{N(N-1) \cdots (N-2t+1)}\right) \right]  \nonumber \\ 
			\nonumber& \leq   \sum_{\alpha \in S_{2t}} \left[\left( (\phi N)^{\vert C_1(\alpha) \vert}  N^{\vert C(\alpha) \vert-\vert C_1(\alpha) \vert} \right) \left( \frac{1}{N(N-1) \cdots (N-2t+1)}\right) \right] \hspace{4cm} (\mbox{from Fact~\ref{fact:unitary_trace_leq_n}})
			\end{flalign*}
			\begin{flalign*}
			& \leq   \sum_{\alpha \in S_{2t}} \left[\left( \phi^{\vert \fix(\alpha) \vert}  N^{\vert C(\alpha) } \right) \left( \frac{1}{(2t)! {N \choose 2t}}\right) \right] & \hspace{2cm} (\mbox{from Lemma~\ref{wg_modsum}})  \\ 	
			& =   \left( \frac{1}{(2t)! {N \choose 2t}}\right) \sum_{i=1}^{2t} \left[  \sum_{ \alpha \in S_{2t} : C(\alpha)=i} \left( \phi^{\vert \fix(\alpha) \vert}  N^{\vert C(\alpha) \vert} \right)  \right]  \nonumber \\ 	
			& \leq   \left( \frac{1}{(2t)! {N \choose 2t}}\right) \left( \sum_{i=1}^{t-1} \left[  \sum_{\alpha \in S_{2t} : C(\alpha)=i}  N^{i}  \right] +\sum_{i=t}^{2t} \left[  \sum_{\alpha \in S_{2t} : C(\alpha)=i} \left( \phi^{ 2i-2t}  N^{ i } \right)  \right]\right) & (\mbox{from Corollary~\ref{upper_bound_on_C}})  \\ 	
			& =   \left( \frac{1}{(2t)! {N \choose 2t}}\right) \left( \sum_{i=1}^{t-1} \left[  \vert \Sigma_{2t-i} \vert  N^{i}  \right] +\sum_{i=t}^{2t} \left[   \vert \Sigma_{2t-i} \vert \left( \phi^{ 2i-2t}  N^{ i } \right)  \right]\right)  \nonumber \\ 	
		& \leq   \left( \frac{1}{(2t)! {N \choose 2t}}\right) \left( \sum_{i=1}^{t-1} \left[ {2t \choose 2}^{2t-i}  N^{i}  \right] +\sum_{i=t}^{2t} \left[  {2t \choose 2}^{2t-i}\left( \phi^{ 2i-2t}  N^{ i } \right)  \right]\right)& (\mbox{from Fact~\ref{ub_on_factorial}})  \\ 
			& =   \left( \frac{{2t \choose 2}^{2t}}{(2t)! {N \choose 2t}}\right) \left( \sum_{i=1}^{t-1} \left[ \left(\frac{N}{{2t \choose 2}}\right)^i \right] + \frac{1}{\phi^{2t}} \sum_{i=t}^{2t} \left[ \left( \frac{ \phi^{ 2}N}{ {2t \choose 2}} \right)^i  \right]\right) \nonumber \\ 	
			& \leq   \left( \frac{ e^{2t-1}(e^2t^2)^{2t}(2t)^{2t}}{(2t)^{2t} (N)^{2t} }\right) \left( \sum_{i=1}^{t-1} \left[ \left(\frac{N}{t^2}\right)^i \right] + \frac{1}{\phi^{2t}} \sum_{i=t}^{2t} \left[ \left( \frac{ \phi^{ 2}N}{ t^2} \right)^i  \right]\right) & (\mbox{from Fact~\ref{ub_on_ncr}})  \\ 	
			& =   \frac{1}{e}\left( \frac{ e^3t^2}{N}\right)^{2t} \left( \sum_{i=1}^{t-1} \left[ \left(\frac{N}{t^2}\right)^i \right] + \frac{1}{\phi^{2t}} \sum_{i=t}^{2t} \left[ \left( \frac{ \phi^{ 2}N}{ t^2} \right)^i  \right]\right) \nonumber \\ 	
		& \leq   \frac{1}{e}\left( \frac{ e^3t^2}{N}\right)^{2t} \left( 2 \left[ \left(\frac{\sqrt{N}}{t}\right)^{2t} \right] + \frac{t}{\phi^{2t}}  \left[ \left( \frac{ \phi^{ 2}N}{ t^2} \right)^{2t} \right]\right) \nonumber \\ 	
			& \leq \clO \left(    \left( \frac{  t \sqrt{N}}{N}\right)^{2t}  + \frac{t}{\phi^{2t}}  \left( \frac{ \phi^2 N}{N}\right)^{2t}  \right)  \nonumber \\ 	
			& \leq \clO \left(    \left( \frac{  t^2}{N}\right)^{t}  + t \phi^{ 2t}  \right).   \nonumber 
		\end{flalign*}

	\end{changemargin}

\end{appendix}

   \end{document}